\documentclass{article}
\usepackage[a4paper, total={6in, 8in}]{geometry}
\usepackage{titlesec}
\usepackage{parskip}
\renewcommand{\baselinestretch}{1.5}  % For 1.5 line spacing

\titleformat{\section}[block]{\normalfont\Large\bfseries}{\thesection}{1em}{}
\titlespacing*{\section}{0pt}{1em}{0.5em}

\usepackage[utf8]{inputenc}
\usepackage{mathtools}
\usepackage{comment}
\usepackage{enumerate}
\usepackage{amsmath}  % For math environments and symbols
\usepackage{amssymb}  % For additional math symbols like \mathbb
\usepackage{amsthm}
\usepackage{thmtools}
\usepackage{thm-restate}
\usepackage{url}

\usepackage{graphicx}
\usepackage{subcaption}
\usepackage{natbib}

\newcommand{\AutoAdjust}[3]{\mathchoice{ \left #1 #2  \right #3}{#1 #2 #3}{#1 #2 #3}{#1 #2 #3} }
\newcommand{\Xcomment}[1]{{}}

\newcommand{\InBrackets}[1]{\AutoAdjust{[}{#1}{]}}% {\left[{#1}\right]}
\newcommand{\Ex}[2][]{\operatorname{\mathbf E}_{#1}\InBrackets{#2}}
\newcommand{\Prx}[2][]{\operatorname{\mathbf{Pr}}_{#1}\InBrackets{#2}}

\newcommand{\dd}{\mathrm{d}}  % for integrals

\newcommand{\eps}{\epsilon}

\newcommand{\E}{\mathbf{E}}

\newcommand{\noaccents}[1]{#1}
\newcommand{\newagentvar}[3][\noaccents]{%
\expandafter\newcommand\expandafter{\csname #2\endcsname}{#1{#3}}%
\expandafter\newcommand\expandafter{\csname #2s\endcsname}{#1{\boldsymbol{#3}}}%
\expandafter\newcommand\expandafter{\csname #2smi\endcsname}[1][i]{#1{\boldsymbol{#3}}_{-##1}}%
\expandafter\newcommand\expandafter{\csname #2i\endcsname}[1][i]{#1{#3}_{##1}}%
\expandafter\newcommand\expandafter{\csname #2ith\endcsname}[1][i]{#1{#3}_{(##1)}}%
}

\newcommand{\newvecagentvar}[3][\noaccents]{%
\expandafter\newcommand\expandafter{\csname #2\endcsname}{#1{\boldsymbol{#3}}}%
\expandafter\newcommand\expandafter{\csname #2s\endcsname}{#1{\boldsymbol{#3}}}%
\expandafter\newcommand\expandafter{\csname #2smi\endcsname}[1][i]{#1{\boldsymbol{#3}}_{-##1}}%
\expandafter\newcommand\expandafter{\csname #2i\endcsname}[1][i]{#1{\boldsymbol{#3}}_{##1}}%
\expandafter\newcommand\expandafter{\csname #2ith\endcsname}[1][i]{#1{#3}_{(##1)}}%
}

\newagentvar{alloc}{x}
\newagentvar{pay}{p}
\newagentvar{val}{v}
\newagentvar{util}{u}
\newagentvar{payment}{p}
\newagentvar{dist}{F}
\newcommand{\ps}{\mathbf{p}}

\newcommand{\avector}[2]{(#1_1,#1_2,\ldots,#1_{#2})}
\DeclareMathOperator{\Rev}{Rev}
\DeclareMathOperator{\CS}{CS}
\DeclareMathOperator{\SW}{SW}

\newcommand{\D}{{\mathcal{D}}}
\newcommand{\DBB}{{D_{\text{BB}}}}
\newcommand{\DNBB}{{D_{\text{NBB}}}}
\newcommand{\pder}[2][]{\frac{\partial#1}{\partial#2}} % derivative convenience

 \newcommand{\note}[1]{}
 \newcommand{\hufu}[1]{}

\newcommand{\Bigcomment}[1]{}

\newtheorem{lemma}{Lemma}

\newtheorem{proposition}{Proposition}

\newtheorem{corollary}{Corollary}
\newtheorem{definition}{Definition}

\begin{document}

\title{Price Stability and Improved Buyer Utility with Presentation Design: A Theoretical Study of the Amazon Buy Box}

% \author{
%   Ophir Friedler\thanks{Outbrain, Netanya, Israel. Email: \texttt{ophirfriedler@gmail.com}}, 
%   Hu Fu\thanks{Shanghai University of Finance and Economics, Key Laboratory of Interdisciplinary Research of Computation and Economics, Ministry of Education, Shanghai, China. Email: \texttt{fuhu@mail.shufe.edu.cn}}, 
%   Anna Karlin\thanks{University of Washington, Allen School of Computer Science \& Engineering, Seattle, WA, USA. Email: \texttt{karlin@cs.washington.edu}}, 
%   Ariana Tang\thanks{Shanghai University of Finance and Economics, Shanghai, China. Email: \texttt{ariana\_tang@outlook.com}}
% }

\author{%
  Ophir Friedler%
    \thanks{Outbrain, Netanya, Israel.  
            Email: \texttt{ophirfriedler@gmail.com}. 
            Ophir Friedler was supported by the European Research Council (ERC), grant No.~337122.}%
  \and
  Hu Fu%
    \thanks{Shanghai University of Finance and Economics;  
            Key Laboratory of Interdisciplinary Research of Computation and Economics,  
            Ministry of Education; Shanghai, China.  
            Email: \texttt{fuhu@mail.shufe.edu.cn}. 
            Hu Fu is supported by the Fundamental Research Funds for the Central Universities of China.  
            Part of this work was done while he was a faculty member at the University of British Columbia and when Ophir visited there.}%
  \and
  Anna Karlin%
    \thanks{Paul G.~Allen School of Computer Science \& Engineering,  
            University of Washington, Seattle, WA, USA.  
            Email: \texttt{karlin@cs.washington.edu}. 
            Anna Karlin is supported by the Air Force Office of Scientific Research  
            (Grant FA9550‑20‑1‑0212) and the National Science Foundation  
            (Grant CCF‑1813135).}%
  \and
  Ariana Tang%
    \thanks{Shanghai University of Finance and Economics, Shanghai, China.  
            Email: \texttt{ariana\_tang@outlook.com}.\\
            \textit{Acknowledgment.} The authors jointly thank Jason Hartline and Alvin Roth for helpful discussions.}%
}

\maketitle

\begin{abstract}
  Platforms design the form of presentation by which sellers are shown to the buyers.  This design not only shapes the buyers' experience but also leads to different market equilibria or dynamics.  One component in this design  is through the platform's mediation of the search frictions experienced by the buyers for different sellers.  We take a model of monopolistic competition and show that, on one hand, when all sellers have the same inspection costs, the market sees no stable price since the sellers always have incentives to undercut each other, and, on the other hand, the platform may stabilize the price by giving prominence to one seller chosen by a carefully designed mechanism.  This calls to mind Amazon's Buy Box.  We study natural mechanisms for choosing the prominent seller, characterize the range of equilibrium prices implementable by them, and find that in certain scenarios the buyers' surplus improves as the search friction increases.
\end{abstract}

% \keywords{Search friction, platform market design, Pandora box problem}

\section{Introduction}
Platforms that enable transactions between buyers and providers have become major venues of the e-commerce. 
As a few salient examples, Amazon and eBay allow buyers to purchase from a vast body of sellers, and Airbnb connects travelers and hosts.  
An important activity on these websites is for one side of the market to search on the other side for a service or good.  
Search is sequential and time-consuming, and usually cannot be exhaustive.  
The way in which the platform \emph{presents} products and services to be searched by the buyers, therefore, crucially affects the market dynamics and outcomes.  
In most platforms, the providers/sellers set their own prices --- third-party sellers on Amazon, for example, set up their own prices, and so do hosts on Airbnb.
In wielding its power to mediate the interface, the platform must take into consideration both the sellers' pricing strategies and the buyers' search policies.

Our starting point is a monopolistic competition model where symmetric but differentiated sellers post their prices visible to all; a buyer must incur an inspection cost to determine her value for any specific seller, and this cost is the same for each seller.\footnote{As we explain later, this is essentially the classical model by \citet{wolinsky86true}, except that in our setting the prices are visible prior to search.  We note that prices usually do not take much effort to see in most online platforms.}
We observe that such a market does not have a pure equilibrium, i.e., no prices are stable, because the sellers are incentivized to undercut each other to be searched first by the buyer.
In practice, price volatility caused by ``algorithmic pricing'' has been voiced by the media \citep{Taft14, brown2016amazon} and has been a concern of consumers.

In this work, we analyze a particularly simple scheme of presentation, wherein \emph{one} seller is made prominent, with drastically reduced inspection cost, while all the other sellers are treated equally.
The platform's rule for choosing the prominent seller sets up a mechanism, in response to which the sellers strategically set their prices.
We show that this presentation scheme, coupled with appropriate mechanisms, can stabilize prices at pure equilibria.
We further analyze the range of prices achievable at equilibria, and derive implications on welfare and consumer surplus.
An interesting discovery, among others, is that an increased inspection cost often increases the consumers' surplus, because the higher search barrier can induce sellers to lower their prices in order to gain prominence.

Before detailing our contributions, we take a moment to introduce Amazon's so-called \emph{Buy Box}, a typical interface design that features one prominent seller.
The Buy Box is an important motivation for our work.
% ; in the rest of the paper we refer to the prominent seller as the one \emph{in the Buy Box}.

\paragraph{The Amazon Buy Box}

\begin{figure}
\includegraphics[width=0.475\textwidth]{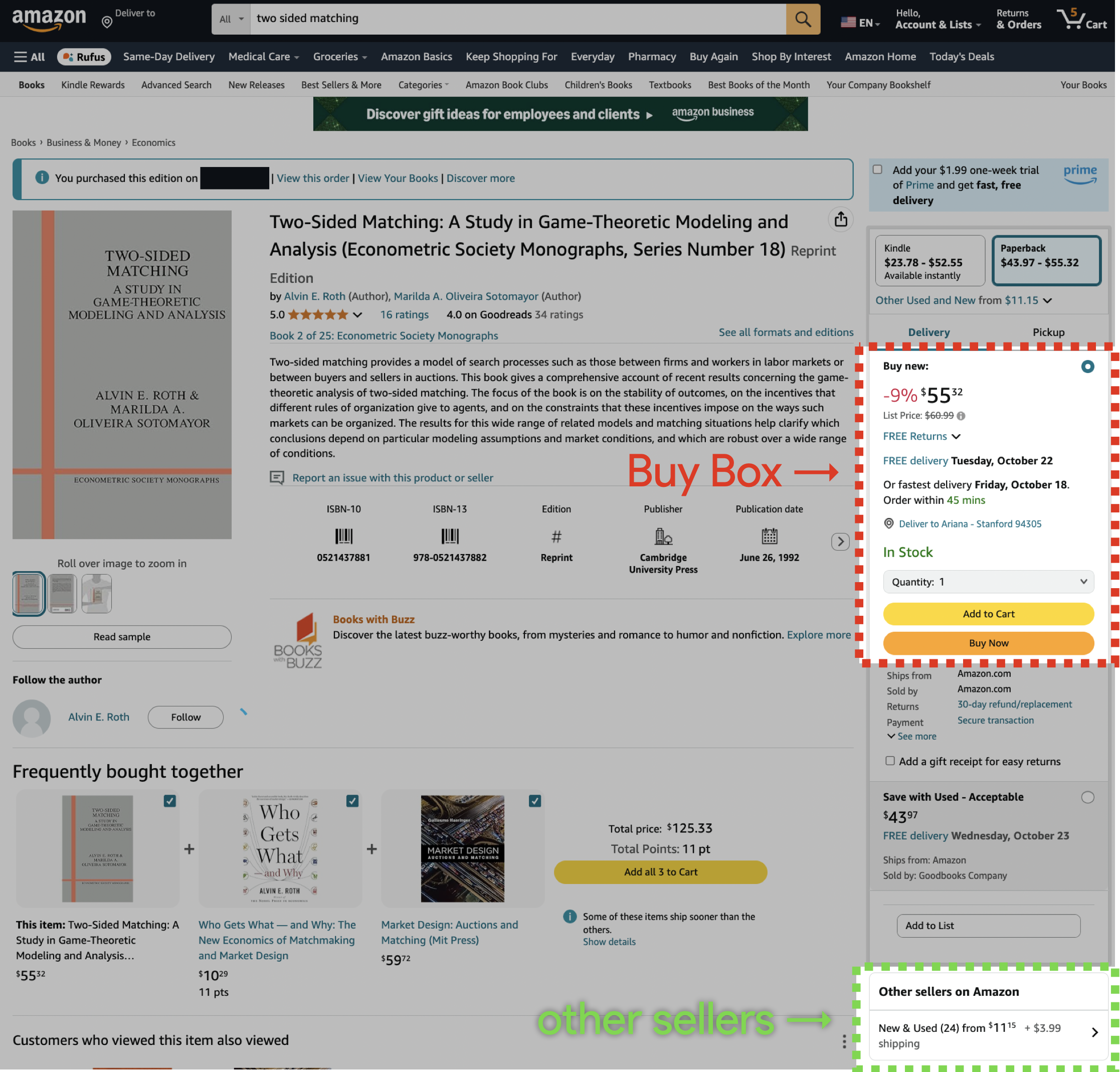}
\includegraphics[width=0.50\textwidth]{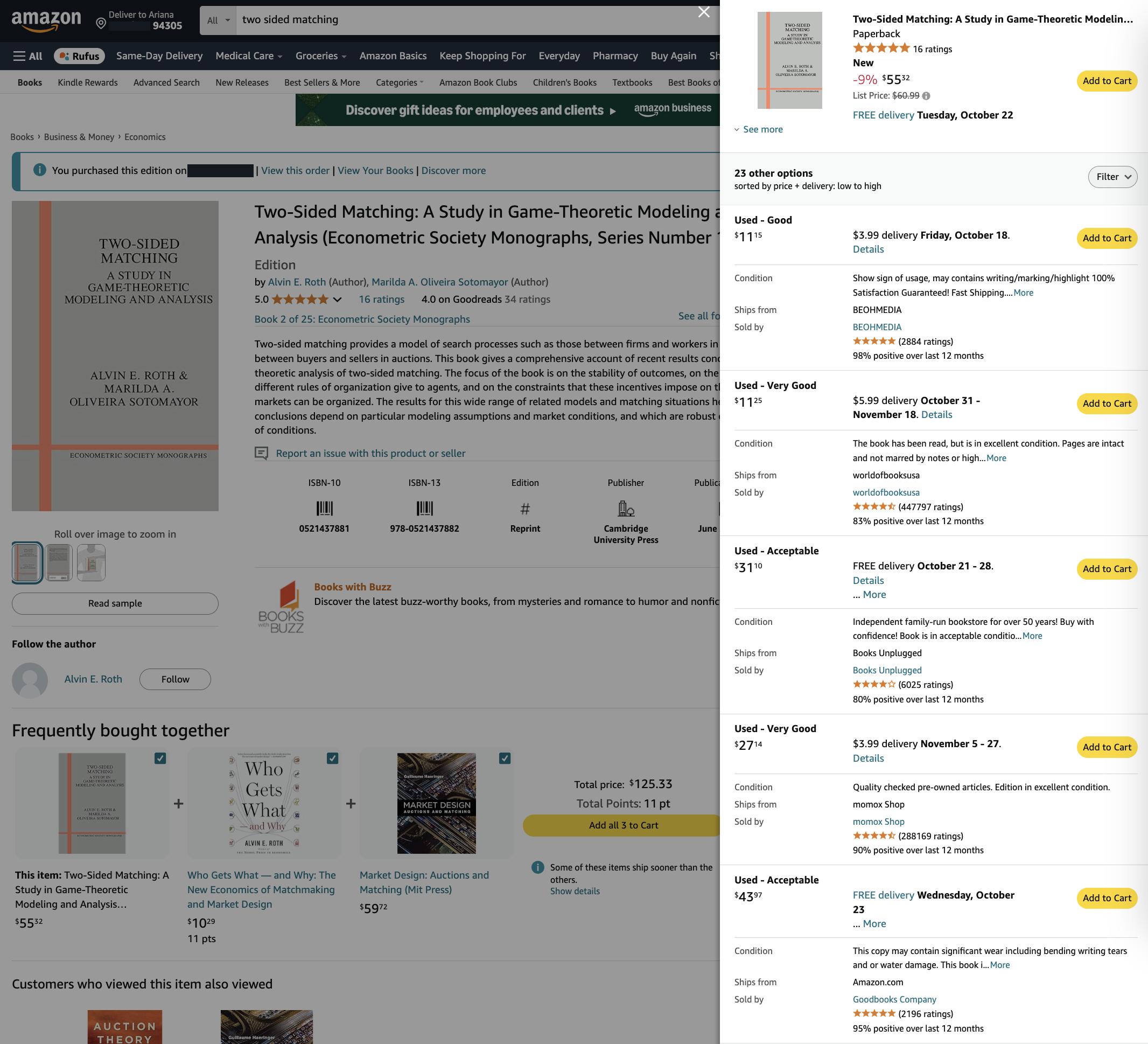}
% \caption{Illustration of Amazon's Buy Box. Top: The Buy Box is highlighted in \textcolor{red}{red}. To view other sellers not in the Buy Box, a buyer needs to scroll down and click on the link we highlight in \textcolor{green}{green}. Bottom: A view of sellers not in the Buy Box.}
\caption{Illustration of Amazon's Buy Box. Top: The Buy Box is highlighted in red. To view other sellers not in the Buy Box, a buyer needs to scroll down and click on the link we highlight in green. Bottom: A view of sellers not in the Buy Box.}
\label{fig:buy_box_demo}
\end{figure}

When a buyer reaches a page on Amazon for a particular product, a {\em Buy Box} is shown (red box in Figure~\ref{fig:buy_box_demo}). 
There are typically multiple sellers for the product, but the Buy Box shows only one seller, with that seller's price and the buttons to press (either ``Add to Cart" or ``Buy now with 1-click") if the buyer were ready to purchase from the Buy Box seller.

% \begin{figure}
%     \centering
%     \includegraphics[width=0.96\linewidth]{figures/buy_box_demo-1.png}
%     \caption{Buy Box interface demonstration. }
%     \label{fig:buy_box_demo-1}
%     \Description{description placeholder}
% \end{figure}
Below the Buy Box (and not always readily visible to a typical buyer), is an additional link (the green box in  Figure~\ref{fig:buy_box_demo}) to the other sellers of the same product (bottom of Figure~\ref{fig:buy_box_demo}), with some information about the prices they offer.
% \begin{figure}
%     \centering
%     \includegraphics[width=0.96\linewidth]{figures/buy_box_demo-2.png}
%     \caption{Outside the Buy Box.}
%     \label{fig:buy_box_demo-2}
%     \Description{description placeholder}
% \end{figure}
% See Figures~\ref{fig:BuyBox}
% and~\ref{fig:NonBuyBox}.

By presenting the sellers in this way, Amazon de facto reduces the inspection cost for the Buy Box seller, increasing the chance of a sale to this seller.
For the other sellers of this same product, Amazon has, for all practical purposes, increased the search cost. 

For the same product, the sellers' offers differ in aspects such as shipping time, location of seller, return and refund policy, product description (especially for used items), and whether they are ``fulfilled by Amazon''.
For any particular seller, a buyer can be certain that her value should be in a certain interval, but must spend effort to determine her exact value, according to her preference.

Placement in the Buy Box has a large impact on the demand for a seller's product, and Amazon's \emph{Buy Box mechanism} for determining the Buy Box seller has drawn attention from both the industry \citep{Bible} and academic works \citep{chen-algorithmic_pricing-2016}.  
The precise mechanism Amazon uses is proprietary information, and is generally considered to depend on many factors, including the sellers' prices, their past performances, etc..
Pricing on Amazon has long been known to involve software that helps to update prices dynamically; at least part of the drive for the incessant price changes is to compete for the Buy Box.

Wild price fluctuation on Amazon has been well documented~\citep{Taft14,brown2016amazon, hackernoon}; \citet{chen-algorithmic_pricing-2016} showed by simulation and by empirical observation that ``sellers that use algorithmic strategies to maintain low prices relative to their competitors are likely to gain a large advantage in the struggle to win the Buy Box''.
With fluctuating prices, not only should the sellers be constantly ready to change their prices, but buyers have also the added burden of uncertainty and should be strategic with the timing of purchase. 

In this work we build a theoretical framework and study a simplified scenario where the Buy Box mechanism only reacts to the prices set by the sellers.
We show that a mechanism which rewards the Buy Box to the most competitive price is intrinsically unstable in that no pure equilibrium exists.
We propose mechanisms admitting pure equilibria, study the range of equilibria prices, and derive welfare and surplus implications.

\paragraph{Our Contributions.}
We start with a model we call \emph{plain presentation} (without a prominent seller). 
 This is a monopolistic competition, where $m$ sellers offer differentiated goods/services at prices set by themselves.  
For each seller~$i$, a buyer has value $v_i$ drawn i.i.d.\@ from a known distribution~$F$.
We normalize values so that $F$ is supported on $[V, V + 1]$ for some $V > 0$; for many results we assume $V \geq 2$, and $c < \Ex{v_i} - V$.
For scenarios such as online retail platforms, these assumptions seem realistic.  (See further discussion in Section~\ref{sec:prelim}.)
The buyer knows~$F$ and can see the prices, but must incur an inspection cost~$c$ to be certain of any particular~$v_i$.
The buyer performs a sequential search, with free recall, in order to maximize her expected utility.
The optimal search policy is given by \citet{weitzman1979optimal}'s renowned index-based policy, which we present in detail in Section~\ref{sec:prelim}.
One feature of sequential search is that, whenever the buyer has found a value high enough, the search stops.
Sellers that are inspected earlier therefore have an advantage.

In this plain presentation model, the buyer's optimal policy goes over the sellers in the order of decreasing prices, which gives the sellers incentives to undercut each other.
We formally show that this market admits no pure symmetric equilibrium.\footnote{\citet{AZ11} mention in passing that ``this framework, where a consumer's match utility is independently distributed across firms, apparently does not lead to a tractable solution for how firms choose prices''.  Our proof may be seen as a formalization of this observation.}
Interestingly, for a market with two sellers, when $V \geq 2$, we show that even asymmetric pure equilibrium does not exist. 
For this result (Theorem~\ref{theorem:no_BB_no_eq}), we use a technique from \citet{armstrongAER}, which allows a nontrivial derivation of the sellers'  respective probabilities of selling had they been at an equilibrium;  a contradiction ensues from this calculation.

  When \emph{prominence} is given to a seller, e.g.\@ in the form of a Buy Box, we assume for simplicity that this seller has no inspection cost, i.e., a buyer always knows her value for this seller before any search, whereas every other seller still incurs an inspection cost~$c$.
As many platforms claim to aim for the lowest, most competitive prices, it is tempting to use a mechanism to induce low prices, e.g.\@ by giving prominence to only the seller with the lowest price.
It is not difficult to see that such an approach gives even more incentive for the sellers to undercut each other, and non-existence of equilibrium from the plain presentation model persists.

We therefore explore mechanisms that admit pure equilibria.
The first mechanism we consider is a \emph{dictator} --- the platform stipulates a price~$t$, and prominence is awarded to only sellers that post price~$t$; if there are multiple such sellers, pick one uniformly at random.
The mechanism is somewhat unnatural, but we show that it has the advantage that it encompasses all equilibria, in the sense that, if any natural mechanism (Definition~\ref{def:standard}) admits a pure symmetric equilibrium, the dictator mechanism admits the same equilibrium (by stipulating that price).
In Section~\ref{sec:implementable}, we completely characterize the range of equilibria realizable in the dictator mechanism (and hence all possible symmetric equilibria in natural mechanisms).  
Our main finding is that the prices implementable at equilibrium always constitute an interval.
The higher the inspection cost~$c$, the more likely a pure equilibrium exists and, when it does, the lower end of the interval is smaller, that is, the lower the implementable lowest price at equilibrium.
For the case of two sellers, we show more crisply that, as $c$ increases, the set of implementable prices grows inclusion-wise.
These findings suggest that, if a platform uses prominence solely to drive down prices, it is effective to increase the inspection cost for non-prominent sellers.

Towards more practical mechanisms, we consider \emph{threshold mechanisms}, which stipulate a price~$t$ and award prominence to sellers that post prices below~$t$; again, if there are multiple such sellers, pick one uniformly at random.
The threshold~$t$ acts similarly to a reserve price in auctions, but  the mechanism does not discriminate among sellers setting prices below~$t$.
This non-discrimination removes sellers' incentives to undercut each other, and makes it possible for pure equilibria to sustain.
If the platform uses a threshold mechanism and aims to drive down prices, it may lower the threshold~$t$ (to the extent that equilibrium still exists), rather than directly encourage competition among the sellers.
We again characterize the range of prices implementable at equilibria in threshold mechanisms (Theorem~\ref{thm:threshold}).
It is perhaps not surprising that there are settings where the dictator mechanism implements a larger range of prices than the threshold mechanism (Proposition~\ref{prop:dict-threshold-gap});
what we find interesting is that, whenever the threshold mechanism admits an equilibrium, the lowest price it is capable of maintaining at equilibrium is the same as that of the dictator mechanism (Corollary~\ref{cor:threshold}).
This suggests that, if the platform aims to drive down prices, the more practical-looking threshold mechanism is just as powerful as the stringent dictator mechanism, as long as the inspection cost is calibrated to a level that allows equilibria to exist.

Finally, we discuss welfare and consumer surplus at equilibria under the prominence mechanisms.
Since the inspection cost is a ``burned'' effort, and prices paid by the buyer and received by the sellers cancel out, the social welfare decreases as $c$ increases (Theorem~\ref{thm:welfare}).
The case of consumer surplus is more complex and interesting.
As $c$ increases, the lowest possible price implementable at equilibrium by a prominence mechanism decreases, but it is also more costly for the consumer to search.
The two factors work in opposite directions for the consumer surplus.
For many value distributions we experiment with, the consumer surplus largely increases as $c$ increases, but it is not necessarily the case that the surplus attains its maximum when $c$ is as large as possible.
In Theorem~\ref{thm:surplus-nonmonotone}, we derive a sufficient condition on~$F$ under which the consumer surplus is maximized at an inspection cost not at its maximum, if the platform implements the lowest equilibrium price.
We see this result as a proof of concept, illustrating that the relationship between the consumer surplus and the inspection cost is a complicated one.
% showing that increased inspection cost may enhance consumer's surplus due to intensified sellers' competition, even though at first glance it costs buyers more effort to acquire information.

\subsection*{Related Works}
The two works most closely related to the current work are \citet{AVZ09} and \citet{AZ11}, as both explicitly considered monopolistic competitions with search frictions and with prominent sellers, but both have key differences from our work.
\citet{AVZ09} consider a market where the a seller's price is invisible to the buyer unless the seller is inspected.
Price changes therefore do not directly affect the buyer's search order, which allows pure equilibria to exist.
We observe that in most online platforms the buyer sees sellers' prices quite easily, and take this as our starting point.
\citet{AZ11} study markets with visible prices, but circumvent the nonexistence of equilibria by considering negatively correlated values, which are stylized in a Hotelling model.  
For online shopping, we consider such negative correlation less well motivated, and stick to the classical i.i.d.\@ setting.
Moreover, we emphasize the mechanism design aspect in assigning prominence, in addition to static equilibrium analysis.
For other works on markets with search friction, we refer the reader to \citet{armstrong17survey} for a beautiful survey.

A small but growing number of works explicitly model buyers' search behavior to study platform presentations.
\citet{Chu_MS2020_PositionRankingAuction}, e.g., consider the ranking of heterogeneous products for a buyer whose search order is determined by this ranking, and they consider a multiobjective optimization that includes surplus and sales revenue.
\citet{Golrezaei_2022_MS_product_ranking} propose a two-stage search policy motivated by empirical evidence, and study product ranking in response to such search. % in order to maximize welfare and purchase probability.
\citet{Branco_MS2012_Search4ProductInfo} propose a continuous inspection procedure for a buyer to evaluate a single seller, and study the seller's optimal pricing in response.
Interestingly, in this setting very different from ours, it is also observed that increased search cost can sometimes benefit the buyer in equilibrium, due to the lower price set by the seller.

% The way we normalize the support of~$F$ to $[V, V + 1]$ and a proof in Section~\ref{sec:plain} both follow \citet{armstrongAER}, although the problem they study is on consumers' private signaling, quite distant from our setting.

% Beyond theoretical studies of search and presentation mechanisms, recent legal and regulatory research has also examined Amazon's Buy Box. For example, 
Amazon's Buy Box has drawn attention from researchers from many angles and disciplines.
Noticeably, legal scholars have studied consequences of algorithmic pricing 
\citep{MacKayWeinstein_DynamicPricingAlgorithms, Cavallo_OnlineCompetitionPricing} and anti-trust concerns of the Buy Box mechanism  \citep{Raval_SelfPreferencingBuyBox, DeMelis_AmazonAntitrustBuyBox}.
Our work takes the perspective of theory in market design, and stands at a distance from these.

% explore how algorithmic pricing affects competition and price adjustments across platforms like Amazon;  study Amazon’s Buy Box mechanism, with a focus on self-preferencing practices and potential antitrust concerns. 
% These works address broader regulatory challenges but provide complementary insights into platform strategies that impact both sellers and consumers.

\subsection*{Organization}
%The rest of the paper is organized as follows. 
In Section \ref{sec:prelim}, we present the monopolistic competition with search frictions, with a description of the buyer's optimal search procedure.
Section \ref{sec:plain} formally proves that the plain presentation sees no stable prices in equilibrium.
Section \ref{sec:buy-box} defines prominence mechanisms and presents the prices implementable at equilibria by the two mechanisms we study.
Section \ref{sec:welfare} analyzes welfare and consumer surplus under prominence mechanisms.
Most proofs are relegated to the appendix.

\section{Model and Preliminaries}
\label{sec:pandoras_box}
\label{sec:prelim}

Consider a unit-demand buyer (she) who looks to buy an item from one of $m$ sellers (them) on a platform.
% Due to horizontally differentiated services, 
The buyer has potentially different values for purchasing from different sellers;
her value $v_i$ for each seller~$i$ is a random variable drawn i.i.d.\@ from a distribution~$F$.
The prior~$F$ is public knowledge to both the buyer and all the sellers, but neither party knows the realization~$v_i$.
 We assume that $F(\cdot)$ is continuous and supported on the normalized range $[V, V + 1]$ for some $V > 0$, an assumption also made by  \citet{armstrongAER}.
 Also following~\citep{armstrongAER}, for much of the paper, we assume $V \geq 2$.\footnote{Since the support is normalized, this assumption generally says that the highest value the buyer may have should not be too much higher than the lowest value.  This can therefore be seen as a tail condition on the value distribution, which may call to mind other tail assumptions made in the mechanism design literature. }
 %, we assume that the support of $F$ be normalized in an interval $[V, V + 1]$ for some $V > 0$.  
Each seller $i$ sets their price~$p_i$,  visible to the buyer.

% Initially, the buyer does not know her exact values for the sellers; 
To learn her value~$v_i$ for seller~$i$, the buyer may pay an \emph{inspection cost}~$c_i$. 
The inspection takes place in a sequential manner.  Namely, at any time, the buyer may inspect the value of a seller of her choice (and incur a cost), or to buy from one of the inspected sellers and quit, or to quit without purchase.
The inspection costs are potentially controlled by the platform, and affect a buyer's behavior.
% The platform's presentation controls the inspection costs, which creates a game among all the three parties.  
% We detail the steps of the game below.
The three parties therefore participate in a game, whose order of actions we detail below.

\paragraph{Order of Actions.}
 The platform, the sellers and the buyer take actions in the following order.
\begin{itemize}
  \item \textbf{Stage 1}: The platform specifies a \textit{mechanism} $\mathcal M$ that maps seller prices, 
  $\mathbf p = (p_1, \ldots, p_m)$,
  to their inspection costs $c_1, \ldots, c_m$.
  \item \textbf{Stage 2}: The sellers, aware of $\mathcal M$ and knowing that the buyer will 
  %screen the sellers and make purchase using  
  perform her optimal search strategy (described below), choose their prices $p_1, \ldots, p_m$.
  \item \textbf{Stage 3}: $\mathcal M$ maps the prices to inspection costs $c_1, \ldots, c_m$.
  \item \textbf{Stage 4}: The buyer performs an optimal sequential search given the value prior $F$, the prices, and the inspection costs.
      % When the buyer enters the platform, she sees for free her exact realized value $v_k$ in the Buy Box. For the sellers not in the Buy–Box, she sees their prices $\mathbf p_{-k}$ but not their values. If the buyer wants to inspect the value for seller $i$, she need to pay a inspection cost of $c$.
\end{itemize}

% In this work we study a limited family of presentations, namely, those choosing one seller and setting their inspection cost to 0; we study mechanisms for such presentations and derive market implications.

\paragraph{Utilities and Objectives.} Each seller~$i$ aims to maximize their expected \emph{revenue}, which is their price $p_i$ times the probability the buyer buys from them.
The buyer aims to maximize her expected payoff, which is the value of the item she eventually buys minus its price, minus all the inspection costs she pays along the way.
Formally, for each $i \in [m]$, let $Z_i$ be the indicator random variable for the event that the buyer inspects seller~$i$, and $Y_i$ the indicator for the buyer purchasing from seller~$i$, then at price profile $\ps = (p_1, \ldots, p_m)$ and inspection costs $c_1, \ldots, c_m$, the buyer's expected utility is
    \begin{align*}
      \CS & := \E \left[\sum_{i \in [m]} Y_i(v_i - p_i)\right] - \E \left[\sum_{i \in [m]}c_i Z_i \right];
    \end{align*}
This is sometimes referred to as the \emph{consumer surplus}.
    The \emph{social welfare} is the sum of the sellers' revenues and the consumer surplus.  Formally,
    \begin{align*}
      \SW & := \Ex{\sum_i (v_i Y_i - c_i Z_i) }.
    \end{align*}

\paragraph{Optimal Sequential Search.}
\label{paragraph:optimal_sequential_search}
 The buyer's optimal sequence of actions is given by the optimal policy for the so-called \emph{Pandora Box} problem, first given by \citet{weitzman1979optimal}, who proved its optimality. 
We now describe this policy, known as the \emph{index policy}. % and refer the reader to \citet{kleinberg2016descending} for a proof of its correctness.
% , and to \citet{BC22survey} for a survey of recent algorithmic developments.

    \begin{enumerate}
    \label{pseudo_algo:Pandora_demo}
        \item For each seller $i$, compute an \textit{index} $\theta_i$, which is the unique solution to the equation:
        $$ \mathbb E_{v_i \sim F} \left[[v_i - p_i]^+ - \theta\right]^+ = c_i,
        $$
        where $[x]^+ := \max(x, 0)$. 
	Tag each seller~$i$ with the index $\theta_i$.
        
        % She initializes by assigning this index value on each seller.
        
        \item % Assume the buyer breaks ties in favor of the inspected sellers. 
	  Among all sellers whose tags are non-negative indices, inspect seller~$i$ with the highest index~$\theta_i$, see the value~$v_i$ and replace the tag by the \textit{utility} of buying from seller~$i$: $v_i - p_i$.
        
	\item  If at any point the largest tag on a seller is a utility (instead of an index), purchase from that seller and quit; otherwise go back to the previous step.
	  If all the sellers have been inspected, then 

	  % continue to inspect the remaining sellers until there is no non-negative index left, then she either
        \begin{itemize}
	  \item if any seller's price is below their value,  make a purchase from the seller that yields the highest utility;
	  \item otherwise leave without any purchase.
        \end{itemize}
    \end{enumerate}

\begin{lemma}[\citet{kleinberg2016descending}]
Following the index policy, the buyer ends up buying from the seller that maximizes $\kappa_i := \min (v_i - p_i, \theta_i)$, if $\max_{i} \kappa_i \geq 0$. 
Her expected utility is $\Ex{\max_i [\kappa_i]^+}$.
\end{lemma}
    % A moment of thought shows that, following This policy is also known as the \emph{Weitzman's algorithm},
    % Henceforth, without loss of generality, we assume that $\theta_i \ge 0, \forall i\in [m]$ for all price profile $\mathbf p$ discussed. 

    % To exclude the degenerate cases where the inspection cost is so high that a buyer never inspects any seller, we assume throughout the paper that when $p_i = 0$, $\theta_i > 0$.

% Without loss of generality assume $\theta_1 \ge \theta_2 \ge \ldots \ge \theta_m$, that for sellers $i < j$, $i$'s product is somehow less expensive or more attractive than $j$ and the buyer should prefer to inspect $i$ first. Since, sellers with negative indices will never be inspected anyway, let $\theta_m \ge 0$. Then, the buyer would implement the above algorithm by using indices $\theta_i$ as thresholds: assume she has inspected the first $k$ sellers, let $k^*= \max_{j\le k}(u_j - p_j)$ and denote her current best option as $u_{k^*} = v_{k^*} - p_{k^*}$. She would pay $c$ to further inspect the next seller $k + 1$ if and only if $v_{k^*} - p_{k^*} < \theta_{k+1}$. Otherwise, she either buys from seller $k^*$ if $u_{k^*\ge 0}$, or leaves without any purchase.

% The interface provided by the platform controls the timing, the cost of the inspection $c$, and the visibility of the prices set by the seller.

\paragraph{Plain Presentation Without Prominence.} 

When the platform presents the sellers ``plainly'', only their prices $\mathbf p:= \avector{p}{m}$ are shown to the buyer.  
Each seller's inspection cost is $c_i = c$, no matter what prices they set.

\paragraph{Presentation with Prominence.} 
In a presentation with prominence, one seller is made prominent, with inspection cost 0, and all the other sellers' inspection costs are set to some $c > 0$.
% The prominent seller is said to be \textit{in the Buy Box}. 
We assume all the sellers' prices are visible to the buyer, but she only sees her value for the prominent seller unless she performs (costly) inspection.
% In effect, the buyer always sees the value of the seller in the Buy Box, before she performs any further inspection among the other sellers (if she chooses to).
The identity of the prominent seller is the result of the platform's mechanism, which we define formally in Section~\ref{sec:buy-box}.

    For a seller not prominent, we often use $\theta_0(c)$ to denote their Weizman index if they post price~$0$, i.e., $\theta_0(c)$ is the solution to the equation $\mathbb E_{v_i \sim F}[v_i - \theta]^+ = c$.  
    When the inspection cost $c$ is clear from the context, we omit it and write $\theta_0$.  

    \noindent \textbf{Assumptions of Non-degeneracy.}
    Throughout the paper we assume $\theta_0 > 0$: when this is not the case, a seller not prominent is never inspected even if they post price~$0$, and the setting degenerates to a monopoly of the prominent seller.
    For most results in the paper, we further assume $\theta_0 > V$.
    When this is not the case, the inspection cost~$c$ is larger than $\Ex{v_i} - V$, which is the uncertainty the buyer has about her value for a seller.
    For one thing, search frictions in online shopping are typically not this big.
    For another, when inspection costs more than the expected uncertainty a buyer has about the value, it is realistic that the buyer may make a purchase \emph{without inspection}, an option interesting by itself \citep{GMS08, doval18, BK19, FLL23, BC23} but not captured by Weitzman's algorithm. 
    Our work focuses on the regime where inspection costs are small relative to the value uncertainty.
    Let $\bar c$ be the cost such that $\theta_0(\bar c) = V$.  
    Assuming $\theta_0 > V$ is equivalent to assuming $c < \bar c$.
    
    A seller with a negative index is never considered by the buyer.
    We therefore mostly restrict our attention to prices low enough to guarantee nonnegative indices.
    For such prices, the calculation of Weitzman indices is simplified by the following proposition:

    \begin{proposition}
    \label{prop:define_theta_0}
     If $\theta_i \geq 0$ under price~$p_i$, then $\theta_i = \theta_0 - p_i$.
    \end{proposition}
    \begin{proof}
        For $\theta_i \geq 0$, $\mathbb E[ [v_i - p_i]^+ - \theta_i]^+ = [v_i - p_i - \theta_i]^+$.
        By definition of~$\theta_i$, $\mathbb E[v_i - p_i - \theta_i]^+ = c$; comparing this with $\mathbb E[v_i - \theta_0]^+ = c$ gives the proposition.
    \end{proof}

\paragraph{Equilibria}
With or without prominence, the expected revenue of any seller~$i$ can be calculated from the price profile: let~$D_i(\mathbf p)$ be the probability with which the buyer buys from seller~$i$ if the price profile is $\mathbf p$, then the seller's revenue is $\Rev_i(\mathbf p) = p_i D_i(\mathbf p)$.  
A price profile $\mathbf{p}$ constitutes a \emph{pure Nash equilibrium} if, for any seller $i \in [m]$, % the following condition holds:
$$
\Rev_i(\mathbf{p}) \geq \Rev_i(p_i', \mathbf{p}_{-i}), \forall p_i'.
$$

\section{Non-Existence of Pure Equilibria in Plain Presentations}
\label{sec:plain}

We show that a plain presentation generally does not admit pure Nash equilibria.  
The argument is most intuitive for the non-existence of symmetric equilibria (Proposition~\ref{proposition_no_symmetric}): any symmetric positive price profile is not stable because the sellers have incentives to undercut each other.
It may come as a surprise that, under mild technical assumptions, even asymmetric equilibria do not exist for two sellers (Theorem~\ref{theorem:no_BB_no_eq}).
% This follows from a nontrivial argument similar to a proof by \citet{armstrongAER}.  
Both results contrast with monopolistic competitions \emph{without} inspection costs, where pure equilibria often exist \citep{Perloff1985_ProductDiffEq}.

 % We show that in a plain presentation without a Buy Box, there is \emph{no pure Nash equilibrium} in the sellers' pricing game. Due to the positive inspection cost, sellers have a perpetual incentive to undercut each other's prices to become the first in the buyer's inspection sequence. This incentive undermines the existence of symmetric equilibria, as any tie in prices (and thus in the buyer's indices for inspection) can be broken by a seller slightly lowering their price to gain precedence and increase revenue. We formalize this intuition in Proposition~\ref{proposition_no_symmetric}, which shows that at least one seller can benefit from unilaterally deviating when prices are tied. Then we extend this result to asymmetric equilibria. Specifically, in the case of two sellers, we prove in Theorem~\ref{theorem:no_BB_no_eq} that no pure-strategy equilibrium exists at all. 

% Consider a plain presentation without Buy Box, where sellers set prices and the buyer employs the Pandora's Box algorithm. A fundamental characteristic of the competition dynamic in this context is, sellers, by strategically adjusting their prices, can significantly manipulate demand so as to increase revenue, particularly when prices are tied:

\begin{restatable}{proposition}{propnotie}
\label{proposition_no_symmetric}
For any $p > 0$, $\ps = (p, \ldots, p)$ is not an equilibrium in a plain presentation.
When $\theta_0 > V$, there is no symmetric equilibrium in a plain presentation.
\end{restatable}

The proposition is intuitive: when every seller posts the same price~$p$, a seller may slightly cut down their price and make sure they are inspected first by the buyer, which produces a jump increase in their demand.
The proof lower bounds this increase, and argues that a small enough undercutting guarantees strict revenue increase.

% We remark that in this proof, we actually do not need the assumption that the support of $F$ is in $[V, V + 1]$ with $V \ge 2$.
 For two sellers, we can show that even \emph{asymmetric} equilibrium does not exist.

\begin{restatable}{theorem}{thmTwoSellerNoEq}
\label{theorem:no_BB_no_eq}
    In a plain presentation with two sellers, whose values are supported on the interval $[V, V + 1]$ for $V \geq 2$, there is no pure Nash equilibrium.
\end{restatable}

% \hufu{maybe elaborate a bit about the proof}\note{shall we put the intuition pic here?}
The proof is technically interesting and involved, and follows the proof strategy of a result by \citet{armstrongAER}, who study consumers' private signaling, a problem quite distant from ours.
For a proof by contradiction, one may assume an asymmetric equilibrium exists and write the equilibrium condition for the two sellers: either seller, if deviating to a different price, should see weakly less revenue.
This gives upper and lower bounds on the demand function of a seller as they vary their price.
We derive Lemma~\ref{lemma:all_buy} below, %similar to one by \citet{armstrongAER}, 
which allows us to express the said bounds and the demand all in terms of one common variable, so that the demand is sandwiched between the two bounds.
% Another crucial observation by \citet{armstrongAER}  (made in a different setting) is that, 
The three functions coincide at the equilibrium point and their derivatives at that point must be equal.
This immediately allows one to express precisely the demand of the sellers at the equilibrium, which leads to a contradiction.
(See Figure~\ref{fig:theorem_1_proof_intuition} for an illustration.)

\begin{figure}
        \centering
        \includegraphics[width=\linewidth]{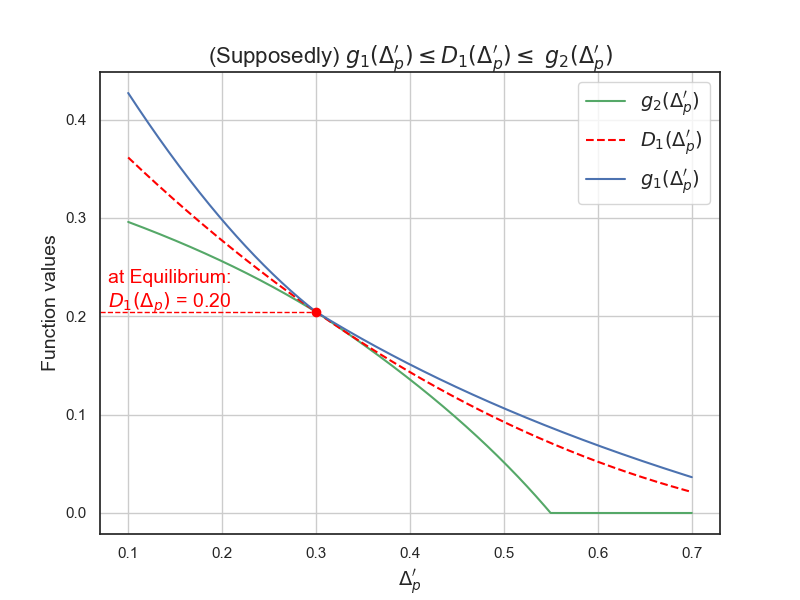}
        \caption[Illustration for a step in the proof of Theorem~\ref{theorem:no_BB_no_eq}.
        ]
        %Equilibrium condition for $(p_1, p_2)$ requires that $g_1(\Delta_p') = \frac{p_1 D_1(\Delta_p)}{p_2 + \Delta_p'}, g_2(\Delta_p') = 1 - \frac{p_2 D_2(\Delta_p)}{p_1' - \Delta_p'}$ \textit{sandwich} $D_1(\Delta_p)$ (inequality~(\ref{ineq:theorem1_key})). Moreover, at $\Delta_p = p_1 - p_2$ the inequality is tight. Therefore, it is further required that the derivative of $g_1(\cdot), g_2(\cdot)$ and $D_1(\cdot)$ coincide at $\Delta_p$.} % The example is demonstrated using the example of $f(v) \propto e^v$.}
        {Illustration for a step in the proof of Theorem~\ref{theorem:no_BB_no_eq}. \par \textmd{
        The equilibrium conditions, together with Lemma~\ref{lemma:all_buy}, require that seller~1's demand, as the seller varies their price, is sandwiched between an upper bound $g_1$ and a lower bound~$g_2$, all three functions expressed in a common variable~$\Delta'_p$.  
        They coincide at the point $\Delta'_p = \Delta_p$, which is the hypothetical equilibrium point. 
        The three functions must have the same derivative at~$\Delta'_p = \Delta_p$.
        }}
        \label{fig:theorem_1_proof_intuition}
        % \Description{placeholder}
    \end{figure}

\begin{restatable}{lemma}{lemmaAllBuy}
    \label{lemma:all_buy}
    In a plain presentation, for $F$ supported on $[V, V+1]$ with $V \geq 2$, the buyer always makes a purchase if the the sellers' prices reach an equilibrium.
\end{restatable}

% The lemma has analogs in later settings we consider as well, but requires more conditions and modifications.
% Henceforth, we assume this outside option is irrelevant and the buyer always purchases from one of the sellers. Perhaps surprisingly, plain presentation without a Buy Box fails to sustain any pure Nash equilibrium for two sellers.

We can further show that even $\epsilon$-approximate \emph{symmetric} equilibrium does not exist in plain presentations with two players.
\begin{restatable}{theorem}{thmeps}
  In a plain presentation, for any $F$ supported on $[V, V + 1]$ and any $c > 0$, there exists a $\Delta > 0$, such that for any positive $\eps < \Delta$, there is no pure symmetric $\epsilon$-equilibrium.
\end{restatable}
% We relegate the proof to the appendix.

\section{Prominence Mechanisms}
\label{sec:buy-box}

As set up in Section~\ref{sec:prelim}, a prominent seller is chosen by a mechanism in response to the price profile.  
This mechanism, called the \emph{prominence mechanism}, is at the crux of the three-party game.
We formally define it below, and quickly observe that the most na\"ive mechanism gives no more price stability than a plain presentation.
% In Section~\ref{sec:dictator}, we introduce a family of powerful mechanism called dictators, and show them to encompass all price equilibria; that is, any symmetric equilibrium implementable by a \emph{standard} mechanism is implementable under a dictator mechanism. Then in section~\ref{sec:implementable}, we give a full characterization of the range of prices implementable at equilibria. In Section~\ref{sec:threshold}, we introduce threshold mechanisms, which are less stringent and more natural looking. We characterize the equilibria implementable in threshold mechanisms, and compare the range with that of dictator mechanisms.

% This section introduces the study of Buy Box mechanism designed to foster market stability and competitive pricing.
% We first show that a seemingly intuitive Buy Box mechanism that prioritizes the lowest-price seller does not support any equilibrium. We then show that appropriate mechanisms for allocating the Buy Box can admit equilibrium, hence stablizing prices in the market. We further derive the mechanism (the dictator-$t$ mechanism) that achieves the lowest equilibrium price, and discuss its welfare and surplus implications. En route, we discuss another natural mechanism (the threshold-$t$ mechanism) and compare the two mechanisms. Towards the end, we endogenize the choice of the inspection cost and discuss its effect on welfare and surplus.

% \paragraph{Prominence Mechanisms}
 % In general, this policy requires a \textit{BuyBox mechanism} for choosing the prioritized seller: 
\begin{definition}
A \emph{prominence mechanism} $\mathcal M$ maps a price profile $\mathbf p = \avector{p}{m}$ to $\mathbf x^{\mathcal M} (\mathbf p):= \avector{x}{m}$, with $\sum x_i \leq 1$, where $x_i$ is the probability that seller~$i$ wins prominence.
\end{definition}
% In this work, we focus on the market outcomes of various Buy Box mechanisms, especially those that induce low prices in pure equilibria.
% One can imagine some natural implementation of Buy Box mechanisms, for example, simply awarding the Buy Box to the sellers that offer the lowest price, or the highest product quality, etc.

% \vskip 12pt
If one aims to drive down prices, it is tempting to simply give prominence to the seller with the lowest price:

\begin{definition}
    The \emph{Lowest Price First (LPF) Mechanism} gives prominence to the seller with the lowest price, breaking ties uniformly at random.
\end{definition}

% Proposition \ref{proposition_no_symmetric} and Theorem \ref{theorem:no_BB_no_eq} naturally extends to the LPF mechanism: the following corollary demonstrate that, despite some action is taken in alleviating the buyer's inspection obstacle by placing the lowest-price seller in a Buy Box, the market yet again admit no pure equilibrium.
It takes only a moment of thought to realize that, compared with the plain presentation, LPF only offers more incentive for the sellers to undercut each other, and can no better sustain equilibria.
% The proof of the following proposition can be found in the appendix.

\begin{restatable}{proposition}{proplpf}% [LPF mechanism is just as bad as plain presentation]
  \begin{enumerate}[(a)]
    \item % The LPF mechanism does not support any pure symmetric equilibrium.
      For any $p> 0$, the symmetric price profile $\mathbf p = (p, \ldots, p)$ is not an equilibrium under the LPF mechanism.
      If $\theta_0 \geq V$, LPF admits no symmetric equilibrium.
    
    \item With two sellers, the LPF mechanism admits no pure equilibrium.
  \end{enumerate}
\end{restatable}

\subsection{The Dictator Mechanisms}
\label{sec:dictator}

% To fix the issue, specifically, consider the following Dictator-$t$ mechanism

We now introduce a family of powerful mechanisms, where a seller has to post a price dictated by the platform to win prominence.
It turns out that these mechanisms encompass all symmetric equilibria in a natural family of mechanisms
(Definition~\ref{def:standard}, Theorem~\ref{thm:dict}). 

\begin{definition} [Dictator-$t$ mechanism]
  The \emph{Dictator-$t$ mechanism} is a mechanism parameterized with a target price~$t > 0$. It assigns prominence to sellers who set their prices at $t$, breaking ties uniformly at random.
  Formally, for any price profile $\mathbf{p} = (p_1, p_2, \ldots, p_m)$, the prominence allocation vector $\mathbf{x}(\mathbf{p})$ is 
    $$
    \mathbf{x}^\text{Dictator-$t$}(\mathbf{p}) = \left(\frac{\mathbb{I}[p_1 = t]}{N_t},\  \frac{\mathbb{I}[p_2 = t]}{N_t}, \ \ldots,\  \frac{\mathbb{I}[p_m = t]}{N_t} \right),
    $$
    where $N_t = \sum_{i \in [m]} \mathbb{I}[p_i = t]$ is the number of sellers pricing at~$t$, and $\mathbb{I}[\cdot]$ is the indicator function.
\end{definition}

% A mechanism can sometimes be designed to be arbitrary and unfair. For instance, assigning the Buy Box to a particular seller regardless of the prices set by others. To ensure that the mechanisms we study are fair in general and foster competition, we focus on the following specific class of nontrivial mechanisms that are \textit{anonymous} and \textit{always allocate}.
% We next show that the Dictator-$t$ mechanism obtains the lowest possible price at equilibrium among a class of nontrivial mechanisms.

\begin{definition}
\label{def:standard}
 A prominence mechanism is said to be
    \begin{itemize}
        \item[(a)] \emph{anonymous}, if for any permutation $\sigma$ on $[m]$ and any price profile $\mathbf{p} := \avector{p}{m}$, the allocation satisfies $x_i(\mathbf{p}) = x_{\sigma(i)}(p_{\sigma(1)}, \dots, p_{\sigma(m)})$ for every $i \in [m]$;
        \item[(b)] \emph{always allocating in equilibrium}, if for any symmetric equilibrium $\mathbf{p}$, $\sum_{i=1}^{m} x_i(\mathbf{p}) = 1$.
    \end{itemize}
    A mechanism is said to be \emph{standard} if it is both anonymous and always allocating in equilibrium.
\end{definition}

\begin{restatable}{theorem}{thmDictatorOpt}
\label{theorem:dictator_opt}
\label{thm:dict}
    For any standard prominence mechanism $\mathcal{M}$, if $\mathbf p = (p, \dots, p)$ is a symmetric equilibrium, then $\mathbf p$ is also an equilibrium for the Dictator-$p$ mechanism.
\end{restatable}

The theorem is proved by observing that, in any standard mechanism, at a symmetric pure equilibrium price~$p$, each seller must be in the Buy Box with probability $\frac 1 m$, the same as in the Dictator-$p$  mechanism.
It remains only to show that a deviation in the dictator mechanism cannot be more profitable than the same deviation in any other standard mechanism.

\subsection{Implementable Prices}
\label{sec:implementable}

In this section we characterize prices implementable at symmetric equilibria by standard mechanisms.
\textbf{Throughout the rest of the paper, we assume $V \geq 2$ and $\theta_0 > V$, i.e., $c < \bar c$.}
We have discussed these conditions in Section~\ref{sec:prelim} (Assumptions of Non-Degeneracy).
% Let $\bar c := \Ex v - V$ be the inspection cost corresponding to $\theta_0 = V$; then equivalently, we assume $c < \bar c$.

% Second, we will show that if the the Buy-Box mechanism set inspection cost higher than $\bar c$, equilibrium price can be made arbitrarily low.

% \begin{restatable}{proposition}{propLowPriceWhenCTooBig}[Arbitrarily low-price equilibrium for $c > \bar{c}$]
%         For any $c > \bar{c}$, there exists a Buy-Box mechanism that supports arbitrarily low equilibrium prices.
% \end{restatable}

\begin{definition}
  A price $t$ is \emph{implementable} if there exists a standard prominence mechanism under which the price profile $\mathbf p = (t, \ldots, t)$ is an equilibrium.
  Let $T(c)$ denote the set of implementable prices when the inspection cost is~$c$.
  % constitutes an equilibrium under some Buy Box mechanism. Given the inspection cost level fixed as $c$, we denote the set of such implementable prices as $$T(c) := \{t : t \text{ is implementable when inspection cost is } c\}.$$
\end{definition}

By Theorem~\ref{thm:dict}, to characterize $T(c)$, it suffices to characterize equilibria under the dictator mechanisms.
$T(c)$, when non-empty, turns out to be always an interval, with endpoinds best expressed as the extreme values of two functions.
The following demand function for deviations is crucial in expressing these endpoints.

\begin{restatable}{lemma}{lemmaDcxExpression}[Demand for a Seller Deviating from the Symmetric Equilibrium under the Dictator Mechanism]
    \label{lemma:D_c_independent_of_t}    
    For \( c \in (0, \bar{c}) \), consider a symmetric equilibrium price \( t \). If a seller deviates to a price \( p \), then their demand, expressed as a function of the amount of deviation $x := p - t$, is  
    % Specifically, the demand after deviating to $p$ is
    % \( \D_c : \mathbb{R} \to [0, 1] \), which maps the price difference \( x = p - t \in \mathbb{R} \) to demand. Specifically, the demand function \( \D_c(x) \) is defined as
    % \begin{align}
    %     \D_c(x) := \begin{cases}
    %         1 
    %         &\text{ if } x \le -1\\
    %         \int_V^{\theta_0(c) -x} \int_{\max(\theta_0(c) + x, v_2 + x)}^{V + 1} \, \dd F(v_1)\, \dd F(v_2) \\
    %         \quad + \int_V^{\theta_0(c) -x} \int_{v_2 + x}^{\max(\theta_0(c) + x, v_2 + x)} F^{m-2}(v_1 - x)\, \dd F(v_1)\, \dd F(v_2) 
    %         & \text{ if } -1 < x < 0\\
    %         \int_V^{\theta_0(c) - x} \left(1 - F(v + x)\right) \, \dd F^{m-1}(v) 
    %         & \text{ if }  0 \le x < \theta_0(c) - V\\
    %         0 & \text{ if } \theta_0(c) - V < x
    %     \end{cases}
    % \end{align}
    \begin{align}
        \D_c(x) := \begin{cases}
	  1, \quad \text{if } x \le -1, \\[8pt]
            \int_V^{\theta_0(c) - x} \left[ \int_{\max(\theta_0(c) + x, v_2 + x)}^{V + 1} 
            \, \dd F(v_1)\, \dd F(v_2) \right.\\
            + \left.\int_{v_2 + x}^{\max(\theta_0(c) + x, v_2 + x)} 
	    F^{m-2}(v_1 - x)\, \dd F(v_1)\,\right] \dd F(v_2), \\
            \quad \quad  \text{if } -1 < x < 0, \\[8pt]
	    \int_V^{\theta_0(c) - x} \left(1 - F(v + x)\right) \, \dd F^{m-1}(v), \\
            \quad \quad \text{if }  0 \le x < \theta_0(c) - V, \\[8pt]
	    0, \quad \text{if } \theta_0(c) - V \leq x.
        \end{cases}
    \end{align}
\end{restatable}

\begin{restatable}{theorem}{thmTc}
    \label{theorem:m_seller_general_T(c)}
    For $c \in (0, \bar c)$, 
    %if $t^*(c) \leq \bar t(c)$, where 
    let $t^*(c) := \sup_{x > 0} \left\{\frac{x \D_c(x)}{\frac1m - \D_c(x)} \right\}$
    \newline
    and 
    $\bar t(c) := \inf_{x < 0} \left\{\left(\frac{(-x) \D_c(x)}{\D_c(x) - \frac1m}\right)^+\right\}\setminus \{0\}$.
    The set of implementable prices $T(c)$, is non-empty if and only if $t^*(c) \leq \bar t(c)$.
    When non-empty, $T(c)$ is the closed interval $[t^*(c), \bar t(c)]$.
    % determined by the inspection cost $c$: %, characterized w.r.t inspection cost $c$:
    % , no standard mechanism admits any symmetric equilibrium.
\end{restatable}

\begin{restatable}{Example}{exampleTc}
    As an example, for the uniform distribution $F = U[2, 3]$ and two sellers, we work out this interval.
    For this distribution, $\bar c = \frac 1 2$.
    For $c \in (0, \bar c )$,
    % we show $t^*(c)$ and $\bar t(c)$ can be calculated in the example of $v\sim \text{U}[2, 3]$ with two sellers. For $c\in (0, \bar c)$, where $\bar c = \frac12$:
\begin{align*}
    t^*(c) 
    & = \frac{\sqrt{-((2\sqrt{2c} - 2c) + \frac12) + \sqrt{2(2\sqrt{2c} - 2c) + \frac14}}}{\frac12\sqrt{2(2\sqrt{2c} - 2c) + \frac14} - 1}\\
    & \quad \times \frac12\left(\frac32 - \sqrt{2(2\sqrt{2c} - 2c) + \frac14}\right)\\
    \bar t(c) & = 2.
\end{align*}
For any $c\in (0, \bar c)$, $t^*(c) <\bar t(c)$. 
Therefore for all small enough $c$, the dictator mechanism admits a pure equilibrium.

% the Buy Box can support any equilibrium price in $[t^*(c), \bar t(c)]$ when inspection cost is $c$.
\end{restatable}

\begin{restatable}{proposition}{propTStarProperties} % [Properties of lowest implementable price]
    \label{prop:t_star_properties}
    % The lowest price implementable by standard mechanism---$t^*(c)$---
    For $c \in (0, \bar c)$, $t^*(c)$ is {continuous} and {monotone decreasing} in~$c$.
    Moreover, % at the end points of $(0, \bar c)$
    \begin{enumerate}[(i)]
        \item $\lim_{c \to \bar c^-}t^*(c) = 0$;
        \label{convergence_0}
      \item $\lim_{c\to 0^+}t^*(c) = t_0$, if a symmetric equilibrium $(t_0, \ldots, t_0)$ exists for the monopolistic competition without inspection cost.
	%(i.e. $c = 0$).
        \label{convergence_bar_c}
    \end{enumerate}
\end{restatable}
     
\begin{restatable}{proposition}{propTBarProperties}% [Properties of highest implementable price]
    \label{prop:t_bar_properties}
    % The highest price implementable by standard mechanism---$\bar t(c)$---
    When there are two sellers, for $c \in (0, \bar c)$, $\bar t(c)$ is continuous and monotone increasing in $c$.
\end{restatable}

\begin{restatable}{corollary}{corrNested}
    \label{corollary:nested}
    For a market with two sellers, if $0 < c < c' < \bar c$, then $T(c) \subseteq T(c')$.
    
    % \note{For $m$ sellers with $m > 2$. $t^*(c)$ monotonically decreases w.r.t. $c$. $\bar t(c)$ doesn't, though.}
\end{restatable}

An extension of Lemma~\ref{lemma:all_buy} to the setting with prominence plays an important role in these characterizations.
% One reason we need the assumption $c < \bar c$ is that this extension is true only under this assumption.

\begin{restatable}{lemma}{lemmaAllBuyBB}
    \label{lemma:lemma_1_extension_for_Buy_Box}
    % Let \( F \) be a distribution with support in \( [V, V + 1] \), where \( V \geq 2 \). 
    If $V \geq 2$ and \( c \leq \bar{c} \), then under any equilibrium of any prominence mechanism, the buyer always makes a purchase.
\end{restatable}

% \paragraph{Remark.} (Comparison of Lemma~\ref{lemma:all_buy} and Lemma~\ref{lemma:lemma_1_extension_for_Buy_Box}) Lemma~\ref{lemma:all_buy} states that in plain presentation without a Buy-Box, for any price equilibrium, the smallest price satisfies $p_m < V$. To obtain a similar conclusion under the Buy-Box mechanism, it is required that the inspection cost $c \leq \bar{c}$ (Lemma~\ref{lemma:lemma_1_extension_for_Buy_Box}). We now show that this assumption 

The assumption $c < \bar c$ is necessary for this extension.
% , and hence the assumption throughout Section~\ref{sec:implementable}.
% The proof of Proposition~\ref{prop:high_price_eq} 
One may construct a prominence mechanism under which prices well above~$V$ are sustained as an equilibrium for sufficiently large~$c$.
% where, with a Buy-Box mechanism accompanied by a sufficiently high inspection cost, a mechanism can be designed such that prices above $V$ (even the monopoly price $p^\star := \arg\max_p \, p \, (1 - F(p))$) may be supported as equilibria (as long as inspection cost is high enough, in this case $c > \mathbb{E}[v] - \frac{p^\star}{m}$):

\begin{restatable}{proposition}{propHighPriceWhenCTooBig}% [High-price equilibrium for $c > \bar{c}$]
    \label{prop:high_price_eq}
    For sufficiently large $c$, there exists a prominence mechanism with an equilibrium at which prices are all higher than~$V$.
\end{restatable}

% However, notice that in these equilibria with $c > \bar{c}$, there is actually no competition in the market: a buyer comes, buys from the Buy-Box without considering the options outside the Buy-Box. Henceforth, we focus on $c \leq \bar{c}$, when the inspection cost is low enough such that it doesn't overshadow potential competition between the sellers. In this case, the buyer would sometimes opt to inspect outside the Buy-Box, compare multiple sellers' products, and buy from the best one.

% We show that, setting any inspection cost $c \in (0, \bar{c})$, the set of prices supportable at equilibrium in the Dictator-$t$ mechanism (and hence any anonymous mechanism that always allocates, by Theorem~\ref{theorem:dictator_opt}) is a closed interval (Proposition~\ref{theorem:m_seller_general_T(c)}). We further fully characterize the left endpoint of this interval, i.e., the lowest price supportable at a symmetric equilibrium by any nontrivial Buy-Box mechanism (Theorem~\ref{theorem:m_seller_general_T(c)}).

\subsection{Threshold Mechanisms}
\label{sec:threshold}

Dictator mechanisms encompass all pure symmetric equilibria, but are unnatural, and seem to abuse the platform's power.
In this section we consider threshold mechanisms, which we consider more practical.
A threshold mechanism only sets a price upper bound for sellers to be eligible for prominence, but otherwise does not discriminate among the eligible sellers.

We again characterize the range of prices implementable at equilibria by the threshold mechanisms (Theorem~\ref{thm:threshold}), and compare it with that of the dictator mechanisms.
An interesting finding is that, even though the threshold mechanisms' range of equilibria prices is generally smaller, and more prone to be empty, whenever an equilibrium does exist, the lowest equilibrium price is the same as that of the dictator mechanisms. 

% We now turn our attention to a deeper examination of the properties associated with various pricing mechanisms. A notable extension of the Dictator mechanism involves relaxing a singular target price, $ t $, perhaps, to a price range. This leads us to consider the twin Threshold-$ t $ mechanism, defined as follows:

\begin{definition}
 For a price $t > 0$, the \emph{Threshold-$t$} mechanism gives prominence to a seller chosen uniformly at random from those whose prices are no higher than $t$. 
 If all prices are higher than~$t$, no seller is given prominence.
\end{definition}

% First, we analyze what equilibrium price can be supported by the Threshold-$t$ mechanism.
For inspection cost~$c$, let $\tilde T(c)$ be the set of prices implementable at symmetric equilibria under a threshold mechanism.  

\begin{restatable}{theorem}{thmThresholdTc}
    \label{theorem:threshold_T(c)}
    \label{thm:threshold}
   % The region $\tilde T(c)$ for Threshold-$t$ mechanism 
    For $ c\in (0, \bar c)$, for $x < 0$,
    define % $\tilde \D_c(x)$:
    \begin{small}
        \begin{align*}
            \tilde \D_c(x) & = \frac1m \D_c(x) + \left(1 - \frac1m \right)\times \\
            & \quad  \left(1 - F(\theta_0(c) + x) + \int_V^{\theta_0(c) - x} (1 - F(v + x))\,\dd F^{m-1}(v)\right).
          \end{align*}
    \end{small}
    Let 
    $\hat t(c) = \sup_{x < 0}\frac{(-x)\tilde \D_c(x)}{\tilde \D_c(x) - \frac1m}$, and $t^*(c)$ the same as in Theorem~\ref{theorem:m_seller_general_T(c)}.
    $\tilde T(c)$ is non-empty if and only if $\hat t(c) \geq t^*(c)$.
    When non-empty, $\tilde T(c)$ is the closed interval $[t^*(c), \hat t(c)]$.
\end{restatable}

The following is immediate from Theorem~\ref{theorem:threshold_T(c)} and Theorem~\ref{theorem:m_seller_general_T(c)}.
\begin{corollary}% 
\label{cor:threshold}
  Given $c\in (0, \bar c)$, if $\tilde T(c) \neq \emptyset$, then the lowest equilibrium price implementable by any standard mechanism is implementable by a threshold mechanism.
  % the Threshold mechanism's equilibrium range in not an empty set, Threshold mechanism attains the lowest implementable equilibrium price.
\end{corollary}

% The following proposition highlights the intrinsic limitations of Threshold-$t$ mechanism and the underlying dynamic of the market. The simplification to a uniform allocation of the Buy Box within a specified price range, while operationally straightforward, is not as powerful as some other more restrictive Buy-Box design.
% The upper endpoint of $\tilde T(c)$ can be strictly smaller than that of $T(c)$.
 The upper endpoint of the range, $\hat t(c)$, is generally strictly smaller than $\bar t(c)$, the upper endpoint of $T(c)$.
% In fact, since the existence of equilibrium in the threshold mechanisms is decided by comparing $\hat t(c)$ with $t^*(c)$, 
We give a concrete example.

\begin{restatable}{proposition}{corrThresholdSucks}
\label{prop:dict-threshold-gap}
  There is an instance where the Threshold-$t$ mechanisms do not encompass all symmetric equilibria.
	% cannot sustain some equilibrium price profile that is achievable by other standard mechanisms.
\end{restatable}

The instance has two sellers, where the value distribution, supported on $[2, 3]$, has pdf  $f(v) \propto e^{v - 2}$. 
In Figure~\ref{fig:dictator_threshold_separation}, we visualize $\tilde T(c)$, the range of prices implementable by threshold mechanisms,
 and compare it with $T(c)$, the range of prices implementable by dictator mechanisms.
    Note that the lowest price implementable by the two mechanisms is the same. 
    In this particular example, both mechanisms admit pure equilibria for any $c \in (0, \bar c)$.
%from all standard Buy Box mechanism.

\begin{figure}
    \centering
    \includegraphics[width=\linewidth]{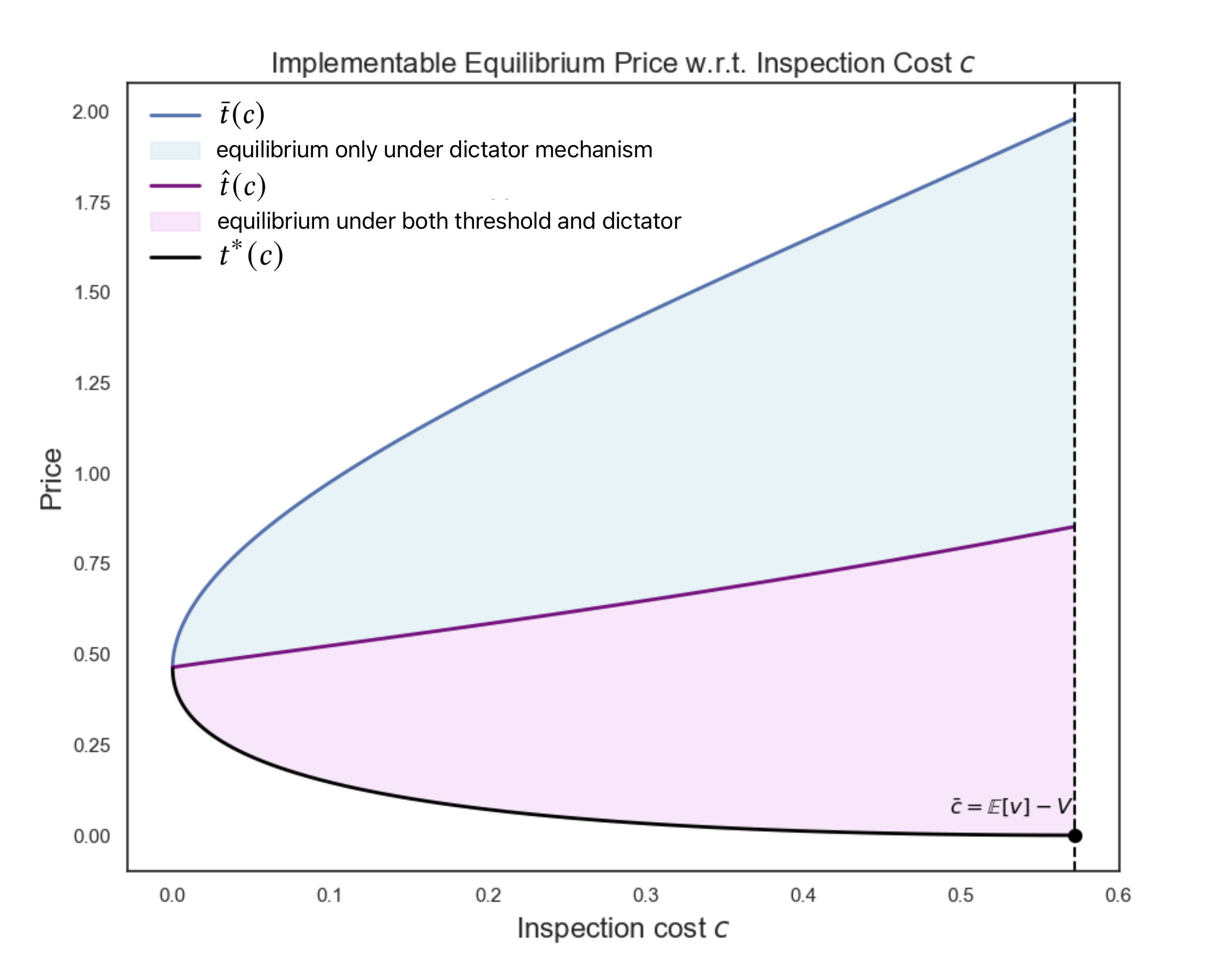}
    \caption[Illustration of implementable prices with threshold mechanisms]{Illustration of implementable prices with threshold mechanisms. \textmd{For two sellers, with $F$ supported on $[2, 3]$ and pdf proportional to $e^{v - 2}$, the magenta area illustrates the range of prices implementable with the threshold mechanism as the inspection cost~$c$ varies.
    The blue area is the range of prices implementable with the dictator mechanism but not with the threshold mechanism.  
    }}
    % we vary inspection cost $c\in (0, \bar c)$ ($x$-axis) and color corresponding implementable regions $T(c), \tilde T(c)$ for Dictator and Threshold mechanism. Both mechanism can sustain equilibrium $\forall c\in (0, \bar c)$, but Dictator mechanism can support higher prices.}}
    % \Description{Illustrative example of the difference between Dictator and Threshold mechanisms under various inspection costs. \newline For each value of inspection cost $c\in (0, \bar c)$, we calculate and plot the upperbound $\bar t(c)$[$\tilde t_1(c)$] for Dictator[Threshold] mechanisms, and their mutual lowerbound $t^*(c)$. We visualize the feasible $t$-region under different inspection cost values.}
\label{fig:dictator_threshold_separation}
\end{figure}

% We'll focus our analysis on the optimal mechanism hereafter.

\section{Analysis of Welfare and Surplus}
\label{sec:SW_CS_analysis}
\label{sec:welfare}

This section analyzes social welfare and consumer surplus at equilibrium under prominence mechanisms.
In particular, we study how changes in inspection cost impact welfare and consumer surplus.

By our assumptions of non-degeneracy and Lemma~\ref{lemma:lemma_1_extension_for_Buy_Box}, the buyer always makes a purchase.
Price change therefore does not directly affect the welfare.
This allows us to show that, at a given inspection cost~$c$, the social welfare does not depend on the price at equilibrium.
% \SWnCSDefinition*

\begin{restatable}{theorem}{SWDeterminedByCnMonotone} % [The property of Social Welfare under symmetric equilibria supported by standard Buy-Box mechanism]
\label{thm:welfare}
  For $c < \bar c$, under any standard prominence mechanism and any symmetric equilibrium price $t$ implemented the mechanism, the social welfare is determined by $c$ only, and monotonically decreases as $c$ increases. The consumer surplus $\CS = \SW - t$.
\end{restatable}

By the theorem, we may write the social welfare at any symmetric equilibrium under a standard mechanism as a function in $c$ only, which we denote by $\SW(c)$.

\begin{corollary}
% [Pareto Frontier in Consumer Surplus and Seller Revenue]
\label{theorem:pareto_frontier}
    % The standard Buy Box mechanism achieves the Pareto Frontier for buyer's and seller's surplus (utility):
    % $$\left\{\left(\SW(c) - t^*(c), t^*(c)\right): \forall c \in [0, \bar c]\right\}.$$
Given $c \in (0, \bar c)$, the highest consumer surplus achievable under any standard prominence mechanisms at a symmetric equilibrium is $\SW(c) - t^*(c)$.
\end{corollary}

For a platform that aims to maximize its customers' surplus, how large the search friction $c$ should it set?  
As $c$ increases, on the one hand, the lowest equilibrium price $t^*(c)$ decreases, which is good for the surplus; on the other hand, search becomes more costly, which hurts the surplus. 
We experimented with many distributions, and saw that, for most of them, the consumer surplus largely increases as $c$ increases for small values of~$c$. 
For some distributions, the surplus increases all the way as $c$ approaches $\bar c$, whereas for others, the surplus takes its maximum at an intermediate value of~$c$, as we illustrate in Figure~\ref{fig:pareto_frontier}.
The next theorem gives a sufficient condition for the latter to be the case.

\begin{figure}
    \centering
    \includegraphics[width=1\linewidth]{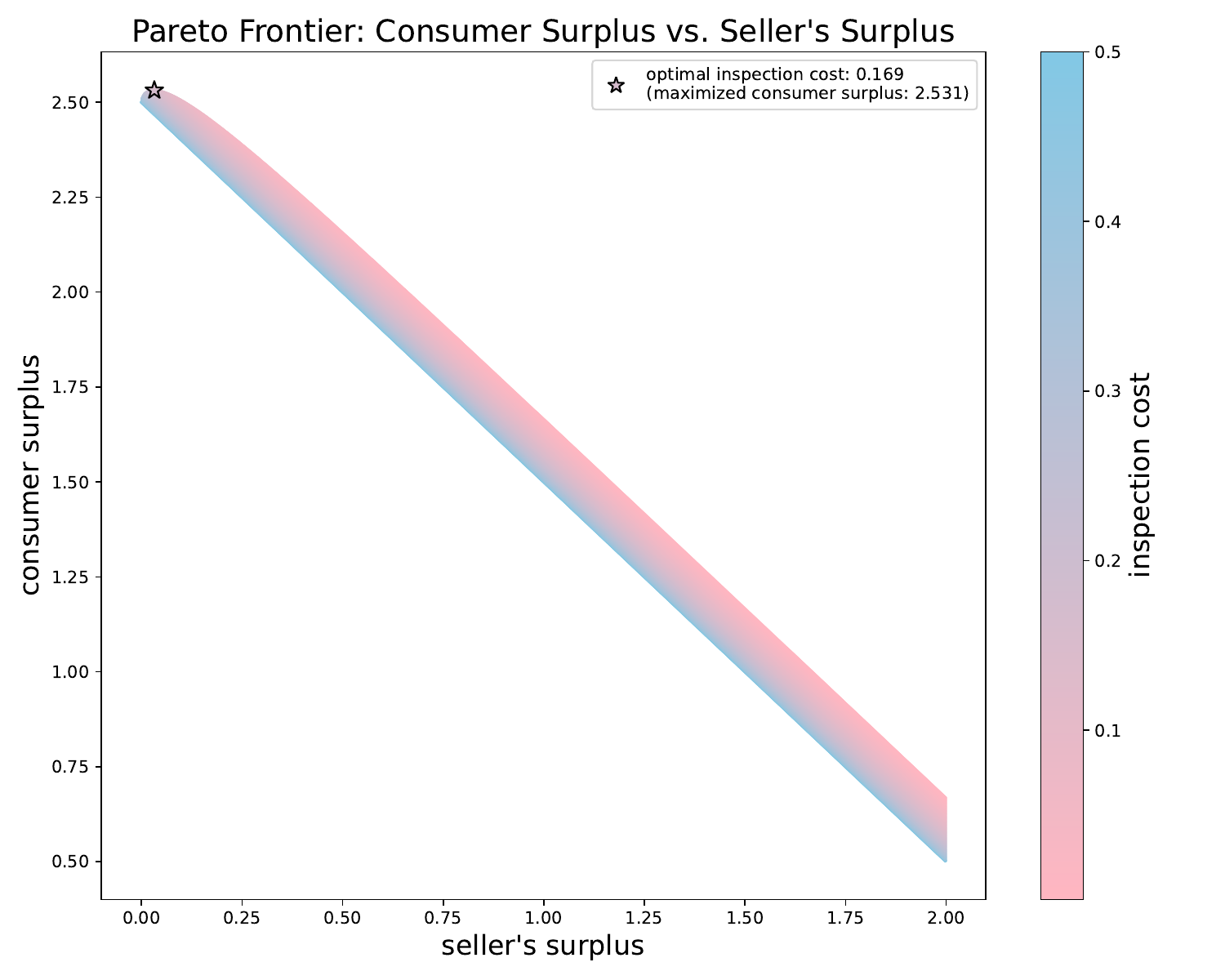}
    \caption[Illustration of achievable consumer surplus and seller revenue at equilibrium.]{Illustration of feasible consumer surplus and seller revenue at equilibrium, colored with corresponding inspection cost.
    \textmd{We plot the region of feasible consumer surplus vs. seller revenues (for two sellers with $F = U[2, 3]$) at equilibrium, and color the region with corresponding inspection cost that admits the equilibrium. The Pareto Frontier is achieved when prices are lowest at $t^*(c)$, where the consumer surplus is maximized.}}
    %As inspection cost $c$ varies, consumer surplus is maximized when price is $t^*(c)$. For the two seller $f(v) \propto e^{v - 2}$ example, we visualize the Pareto curve as inspection cost $c$ varies within $(0, \bar c)$.}}
    \label{fig:pareto_frontier}
    % \Description{Placeholder}
\end{figure}

% As illustrated in Figure \ref{fig:pareto_frontier}, the Pareto frontier reflects the trade-off between consumer surplus and seller revenue, under standard Buy-Box mechanism with varying inspection costs. Specifically, consumer surplus can be optimized when the Buy-Box designer sets inspection cost to be either at an intermediate level $c^* \in (0, \mathbb E[v] - V)$ or at the maximum lower-bound (for index) $c^* = \mathbb E[v] - V$, and the corresponding lowest price achievable $t^*(c)$.

% For $m = 2$ sellers, we demonstrate a sufficient condition for the consumer surplus to be optimized at an intermediate level of inspection cost is $f(V^+) < \frac{1}{13}f'(V^+)$.
\begin{restatable}{theorem}{propCSsufficientcondition}
\label{thm:surplus-nonmonotone}
        For two sellers, a sufficient condition for the consumer surplus to be optimized at an intermediate level of inspection cost is $f(V^+) < \frac{1}{13}f'(V^+)$.
\end{restatable}

The proof is rather technical.  
One step that makes use of there being only two sellers is the following relatively succinct expression for the social welfare.

\begin{restatable}{lemma}{thSWclosedform}
    \label{lemma:SW_closed_form}
        Under standard prominence mechanism with inspection cost $c\in (0, \bar c)$, at any symmetric equilibrium, for $m = 2$ sellers, % social welfare can be written the following closed-form presentation.
        \begin{align}
            \label{social_welfare}
	    \SW(c) & = \int_V^{\theta_0(c)} F(s)(1 - F(s))\, \dd s + \mathbb{E}[v].% , \quad \forall c \in (0, \bar{c}]
        \end{align}
\end{restatable}

% An upper bound of inspection cost is $\bar{c} = \mathbb{E}[v] - V$: that any further increase in inspection cost renders the market outside the Buy Box irrelevant, and social welfare essentially reflects the market as if it contained only one seller. As inspection costs drops, buyers become more inclined to explore beyond the Buy Box, enhancing product match quality and offsetting the loss incurred by the inspection cost. As $c \to 0$, social welfare rises to $2\mathbb{E}[\max_i v_i]$.

\section{Conclusion}
This work builds on a classical model of monopolistic competition with search friction and studies market outcomes when a prominent position is allotted to a seller by a mechanism.
We find that visibility of prices, in presence of search frictions, may cause price fluctuation, which can be exacerbated by a mechanism that unobstructedly encourages competition by rewarding prominence.
We show that properly designed mechanisms for prominence can stabilize prices.  
We analyze the range of prices implementable at equilibrium, and propose the threshold mechanism which is detail free, easy to implement, and moderately encouraging competition.
With the sellers competing for prominence, the consumer surplus may sometimes increase with higher search friction, the disutility of search being offset by the price decrease due to heightened competition.

Amazon's Buy Box is a salient example of platform-mediated prominence mechanisms and largely motivates this work.  
We recognize that the Buy Box mechanism in practice is far more complicated, and reacts to much more information than prices set by the sellers.
This work takes an analytical approach towards a theoretical understanding: keeping most factors unchanged or symmetric, singling out one factor, that of prices, and studying its behavior in response to the prominence mechanism.

Even though, to derive our analytic results, we specify the buyer's search policy (Weitzman's index-based policy in this case), we believe that certain insights we obtain do not rely on this particular policy.
For instance, any policy that searches in the order of decreasing price and does not necessarily exhaust all sellers, is likely to incentivize sellers to undercut each other in order to be searched early.

Prominence in presentation, exemplified by the Amazon Buy Box, is a design simpler than many other forms of presentation, and arguably lends itself more readily for modeling and analysis.  
We see our work as taking a step towards a fuller understanding of its design principles.

% \begin{acks}
%     The first author was supported by the \grantsponsor{erc}{European Research Council}{https://erc.europa.eu/homepage} 
%     under grant agreement number~\grantnum{erc}{337122}.
%     The second author was supported by the \grantsponsor{frf}{Fundamental Research Funds for the Central Universities}{https://www.gov.cn/zhengce/zhengceku/2021-12/14/content_5660674.htm} of China.
%     The third author was supported by the \grantsponsor{airforce}{Air Force Office of Scientific Research}{https://www.afrl.af.mil/AFOSR/} under Grant No.~\grantnum{airforce}{FA9550-20-1-0212} and the \grantsponsor{nsf}{NSF}{https://www.nsf.gov/} under Grant No.~\grantnum{nsf}{CCF-1813135}.
%     Part of the work was done when the second author was a faculty member at University of British Columbia, and when the first author was visiting there.
%     The authors thank Jason Hartline for helpful discussions.
% \end{acks}

\newpage

\bibliographystyle{ACM-Reference-Format}
% \bibliography{bib,new_papers}
\bibliography{camera_ready/bib,camera_ready/new_papers}

\appendix
\section{Omitted Proofs from Section~\ref{sec:plain}}

\propnotie*

\begin{proof}
Consider any symmetric price profile $\mathbf p = (p, \ldots, p)$.

If $p > 0$, we show that there exists a $p' < p$ to which some seller may deviate and improve their revenue.
Let $\theta$ be the Weitzman index of a seller posting price~$p$.
\label{proof:defA_proposition_no_symmetric}
Let $A$ be the event that there are at least two sellers from whom the buyer has value at least $\theta + p$.
For $c > 0$, $\Pr[A] > 0$.
When $A$ happens, the index algorithm always leads the buyer to buy from such a seller.
No matter how the tie is broken, there exists a seller~$i$ such that, the probability that the buyer buys from~$i$ when $A$ happens is at most $\Prx A / m$.
On the other hand, let $B$ be the event that $v_i \geq \theta + p$. 
Then $\Prx{B \cap A} > \Prx A / m$.
If $i$ deviates to any $p' < p$, with a new index $\theta' > \theta$, then whenever $B$ happens, the buyer buys from~$i$.
In particular, when both $A$ and~$B$ happens, the buyer buys from~$i$.
When $v_i < \theta + p$, the buyer buys from seller~$i$ with a probability no less than before the deviation.
Therefore, for any $p' < p$, the demand from seller~$i$ jump increases by at least $\Delta := \Prx{B \cap A} - \frac{\Prx{A}} m > 0$.
Let $D$ be the demand from seller~$i$ before the deviation.
There must exist $p' \in (0, p)$ such that $p'(D + \Delta) > pD$.
Such a $p'$ is a profitable deviation for seller~$i$.

For $p = 0$ and $\theta_0 > V$, % all sellers' Weitzman indices are $\theta_0 > 0$ by our assumption of non-degeneracy. 
there is a small enough $p' > 0$ such that the index $\theta'$ corresponding to $p'$ is still strictly positive.
A seller~$i$ deviating to~$p'$ sells their product when $\kappa_i > \kappa_j$ for all $j \neq i$.
This is true at least when $v_i - p' \geq \theta'$ (in which case $\kappa_i = \theta'$) and $v_j < \theta' < \theta_0$ (in which case $\kappa_j = v_j$).
So seller $i$'s demand is at least $(1 - F(\theta' + p'))F^{m - 1}(\theta') = (1 - F(\theta_0))F^{m-1}(\theta')$.
If $\theta_0 > V$, $\theta'$ can be taken to be greater than~$V$, so that this demand is strictly positive, in which case the deviation to~$p'$ is profitable for seller~$i$.
\end{proof}

\thmTwoSellerNoEq*

\begin{proof}
    By Proposition~\ref{proposition_no_symmetric}, we only need to further show that asymmetric equilibria do not exist.
    For the sake of contradiction, suppose $(p_1, p_2)$, with $p_1 > p_2$, is an equilibrium.
    Then, $\theta_2 > \theta_1 \geq 0$.
    It is straightforward to see that the buyer buys from seller~1 with probability at most $\frac{1}{2}$.
    We show that $(p_1, p_2)$ being an equilibrium would contradict this fact.
    
    The buyer first inspects seller~$2$, after which, she only inspects seller~1 if $v_2 - p_2 < \theta_1 = \theta_0 - p_1$ (Proposition~\ref{prop:define_theta_0}). Lastly, after inspecting seller~$1$, she buys from seller~1 if $v_2 < v_1 - (p_1 - p_2)$.
    (Recall from Lemma~\ref{lemma:all_buy} that the buyer must buy from one of the sellers.)
    Therefore, the demands of the two sellers are determined by the difference of their prices, $\Delta_p := p_1 - p_2$:
    \begin{align*}
        D_1(\Delta_p) & := \int_{V}^{\theta_0 - \Delta_p} \left[1 - F(v + \Delta_p)\right]\ \dd F(v), \\
        D_2(\Delta_p) & := 1 - D_1(\Delta_p).
    \end{align*}
    
    Consider small unilateral deviations by the two sellers, respectively.
    Let $\Delta_p'$ be close but not equal to $\Delta_p$.
    % (specifically, $\Delta_p' \in N_0(\Delta_p)$, a deleted neighborhood of $\Delta_p$).
    Let $p_1' = p_2 + \Delta_p'$ and $p_2' = p_2$ (for seller~1's deviation), and $p_2' = p_1 - \Delta_p'$ and $p_1' = p_1$ (for seller~2's deviation). 
    When seller~$1$ deviates, $p_2$ remains $p_2 < 2$ by Lemma~\ref{lemma:all_buy}. 
    When seller~2 deviates to $p_2' = p_2 - (\Delta_p' - \Delta_p)$, as long as $\Delta_p'$ is close enough to $\Delta_p$, $p_2' < 2$ still holds. 
    Therefore, the buyer is still guaranteed to always make a purchase after either of these deviations. 
    Therefore, the demands for both sellers' can still be expressed in terms of the price difference $\Delta_p'$:
    \begin{align*}
        D_1(\Delta_p') & := \int_{V}^{\theta_0 - \Delta_p'} \left[1 - F(v + \Delta_p')\right]\ \dd F(v), \\
        D_2(\Delta_p') & := 1 - D_1(\Delta_p').
    \end{align*}
    
    % Assuming that $D_1(\cdot)$ is continuous in some small neighborhood of $\Delta_p$, and 
    Since $(p_1, p_2)$ is assumed to be an equilibrium, we have:
    \begin{align*}
        \text{Seller 1:}\quad & p_1' D_1(\Delta_p') = (p_2 + \Delta_p')D_1(\Delta_p') \leq p_1 D_1(\Delta_p);\\ % \quad \forall\ \Delta_p' \in N_0(\Delta_p); \\
        \text{Seller 2:}\quad & p_2' D_2(\Delta_p') = (p_1 - \Delta_p')D_2(\Delta_p') \leq p_2 D_2(\Delta_p). % \quad \forall\ \Delta_p' \in N_0(\Delta_p).
    \end{align*}
    Since $D_1(\Delta_p') + D_2(\Delta_p') \equiv 1$, we can combine these inequalities:
    \begin{align}
        \label{ineq:theorem1_key}
        1 - \frac{p_2 D_2(\Delta_p)}{p_2'} \leq D_1(\Delta_p') \leq \frac{p_1 D_1(\Delta_p)}{p_1'}.
        % \forall \Delta_p \in N_0(\Delta_p).
    \end{align}
 % Taking limits as $\Delta_p' \to \Delta_p$, 
 % we obtain:
    %\[ \lim_{\Delta_p' \to \Delta_p} \left(1 - \frac{p_2 D_2(\Delta_p)}{p_2'}\right) = 1 - \frac{p_2 D_2(\Delta_p)}{p_2} = D_1(\Delta_p), \]
    % and
    % \[ \lim_{\Delta_p' \to \Delta_p} \left(\frac{p_1 D_1(\Delta_p)}{p_1'}\right) = \frac{p_1 D_1(\Delta_p)}{p_1} = D_1(\Delta_p). \]
    As $\Delta_p'$ approaches~$\Delta_p$, the left hand side and the right hand side both approach $\D_1(\Delta_p)$ (Figure~\ref{fig:theorem_1_proof_intuition}).
    For these inequalities to hold in a neighborhood of $\Delta_p$, the derivatives of the left-hand side and right-hand side with respect to $\Delta_p'$ at $\Delta_p$ must be equal, implying:
    \[
        - \frac{p_2 D_2(\Delta_p)}{(p_2)^2} = - \frac{p_1 D_1(\Delta_p)}{(p_1)^2}.
    \]
    Solving this along with $D_1(\Delta_p) + D_2(\Delta_p) = 1$ (Lemma~\ref{lemma:all_buy}), we find that $D_1(\Delta_p) = \dfrac{p_1}{p_1 + p_2} > \dfrac{1}{2}$ (since $p_1 > p_2$ by assumption), which contradicts the earlier observation that $D_1(\Delta_p) \leq \dfrac{1}{2}$.
    
    % Alternatively, if $D_1(\cdot)$ is not continuous around $\Delta_p$, then there exists a small deviation $\Delta_p'$ such that $|D_1(\Delta_p') - D_1(\Delta_p)| > \delta$ for some $\delta > 0$. In this case, one of the sellers can profitably deviate by adjusting their price slightly, contradicting the assumption that $(p_1, p_2)$ is an equilibrium.
    
    Therefore, $(p_1, p_2)$ with $p_1 > p_2$ cannot be an equilibrium. Overall, under plain presentation, no pure equilibrium exists in the presence of inspection costs for two i.i.d.\ sellers.
\end{proof}

\lemmaAllBuy*
\begin{proof}
  Suppose $\mathbf p := \avector{p}{m}$ is an equilibrium. 
  Without loss of generality, suppose $p_1\ge p_2\ge \ldots \ge p_m$, and $0 \leq \theta_1 \le \ldots \le \theta_m$. 
  % Discuss over whether $\mathbf p$ contain zeros:
  \begin{itemize}
      \item If there is some $p_i = 0$, then $p_m = 0$. 
      and $\kappa_m \geq 0$ by our assumption on the inspection cost.
      % $= \min(\theta_m, v_m - p_m)\ge \min(\theta_m, v_m)\ge 0$. 
      In this case, the buyer would at least inspect seller $m$ and obtain  positive utility $v_m - p_m = v_m \ge V$ if purchasing from seller $m$. 
      % So she is guaranteed to purchase.
      \item If $\forall i, p_i > 0$, notice that any seller $i$ should have strictly positive revenue~$\Rev_i(\mathbf p) > 0$ in an equilibrium. 
      This is because, fixing all others' prices, seller~$i$ can always post a positive price smaller than all others' to be inspected first and to guarantee some revenue.
      % $\mathbf p_{-i}$, seller $i$ can always price slightly below others' lowest price~$p_i' = \min_{j\ne i} p_j - \epsilon > 0$ to obtain positive demand $D_i^{\textbf{null}}(p_i', \mathbf p_{-i})$ hence positive revenue. Therefore, for an equilibrium pricing profile $\mathbf p$, all seller's revenue would be strictly positive.
      
      Now,  any seller $i$ must have $p_i \leq p_m + 1$, otherwise a buyer always prefers seller $m$ to~$i$ (since $v_i \leq v_m + 1$ with probability~1, and seller $m$ is inspected before $i$), and $i$'s demand must be zero. 
      On the other hand, if any seller deviates to a price lower than $p_m - 1$, he would get all the demand. 
      In an equilibrium, since any seller cannot be better off with this deviation, we have $p_i D_i(\mathbf p) \ge p_m - 1, \forall i\in [m]$. Summing  these inequalities over~$i$, we have:
        \begin{align*}
            m & \ge m p_m - \sum_{i\in [m]} p_i D_i(\mathbf p)\\
            & \ge m p_m - \sum_{i\in [m]} p_1 D_i(\mathbf p) \\
            & \ge m p_m - p_1  
            \ge m p_m - 1 - p_m,
        \end{align*}
     which gives $p_m \le \frac{m + 1}{m - 1} < 2$. 
     Therefore, the buyer at least finds seller $m$'s product desirable (i.e. $v_m - p_m > 0, \forall v_m \in [V, V + 1]$ for $V \ge 2$), and must make a purchase from some seller.
    % anyway (although not directly buying from seller $m$). 
  \end{itemize}
  %   This concludes our proof.
\end{proof}

\thmeps*
    \begin{proof}
    Consider symmetric price profile $(p, p, \ldots, p)$, we will show that this cannot be any $\epsilon$-equilibrium for any $\epsilon \le \Delta$. $\Delta$ is given by
    $$
    \Delta = \min\left(\frac1m r(x^*), (m-1)r(x^*)\left(\Pr[B_0\cap A_0] - \frac{\Pr[A_0]}m\right)\right)
    $$
    where
    \begin{itemize}
        \item $r(x)$ as will be shown later, corresponds to seller $i$'s post-deviate demand:
        $$r(x) := x \int_{V}^{\theta_0 - x} (1 - F(v + x)) \, \dd F^{m-1}(v).$$
        And, $x^* = \arg\max_{x \ge 0} r(x)$.
        \item $A_0$ is the event that there are at least two sellers from whom the buyer has value at least $\theta_0$. $B_0$ is the event that $v_i \ge \theta_0$, for a given seller $i$.
    \end{itemize}
    We will now show that any price $p$ cannot support symmetric $\epsilon$-equilibrium $(p, \ldots, p)$ for $\epsilon < \Delta$. Discuss over the value of $p$:
    \begin{itemize}
        \item Case: $p\le (m-1)r(x^*)$:
        
        When seller $i$ deviates to $p_i > p$, she will be inspected last in plain presentation. Her demand is
            \begin{align*}
                & D_i(p_i; p_{-i} = p) \\
                = &  \int_{V}^{\theta_0 - (p_i - p)} (1 - F(v + p_i - p)) \, \dd F^{m-1}(v).
            \end{align*}
            Take $x = p_i - p, x > 0$, notice that 
            \begin{align*}
                r(x) = (p_i - p)D_i(p_i; p_{-i} = p)
            \end{align*}
            Observe that $x^* = \arg\max_{x \ge 0} r(x) > 0$ because first, $r(x) < 1$. And at least if we take, say $\hat x = V + \frac{\theta_0}2$, we have
            \begin{align*}
                r(\hat x) & = \hat x\int_{V}^{V + \frac{\theta_0}2} (1 - F(v + \hat x)) \, \dd F^{m-1}(v)\\
                & \ge \hat x \int_{V}^{V + \frac{\theta_0}2} (1 - F(\theta_0)) \, \dd F^{m-1}(v)\\
                & = \left(V + \frac{\theta_0} 2\right) \frac{\theta_0}2(1 - F(\theta_0)) F^{m-1}\left(V + \frac{\theta_0} 2\right)\\
                & > 0 = r(0).
            \end{align*}
        So, it's safe to conclude that $x^* > 0$ and $r(x^*) > 0$.
        
        At symmetric price $p$, seller $i$ can deviates to $p_i = p + x^*$ to obtain at least $\Delta$-more revenue: (notice that, at symmetric price $p$, every seller's revenue is less than $\frac 1m p$)
        \begin{align*}
            \Rev_i\left(p, \ldots, p\right) + \Delta & \le \frac1m p + \Delta\\
            & \le \frac{m-1}m r(x^*) + \frac1m r(x^*)\\
            & = r(x^*) \quad \text{(let $p_i = x^* + p$)}\\
            & = (p_i - p) D_i(p_i; p_{-i} = p)\\
            & \le p_i D_i(p_i; p_{-i} = p).
        \end{align*}
        In other words, this breaks the $\epsilon$-equilibrium for $\epsilon < \Delta$.
        \item Case: $p > (m-1) r(x^*)$ and $\theta \le 0$:
        for price profile $(p, \ldots, p)$, all seller's revenue would be zero because the buyer does not even care to inspect any seller. For seller $i$, by posting price $p_i = V$, by assumption $\theta_0 > V$, her index becomes positive:
            $$
            \mathbb E[(v_i - V) - 0]^+ \ge \mathbb E[v_i - \theta_0]^+ = c.
            $$
            Then, every buyer would inspect and buy from seller $i$, so that seller $i$'s revenue rises to $V$. So this deviate breaks the $\epsilon$-equilibrium.
            
        \item Case: $p > (m-1) r(x^*)$ and $\theta > 0$:
        
        Similar as in proof~\ref{proof:defA_proposition_no_symmetric} of Proposition~\ref{proposition_no_symmetric}, let $A$ be the event that there are at least two sellers from whom the buyer has value at least $\theta + p$. Let $B$ be the event that $v_i \ge \theta + p$.
        By Proposition~\ref{prop:define_theta_0}, $\theta = \theta_0 - p$. So $A_0 = A$ and $B_0 = B$.
        
        As demonstrated earlier in in proof~\ref{proof:defA_proposition_no_symmetric} of Proposition~\ref{proposition_no_symmetric}, the demand for seller $i$ when he deviates to some $p_i < p$, fix all other $p_{-i} = (p, \ldots, p)$, increases at least by $\Pr[B\cap A] - \frac{\Pr[A]}m$. Assume seller $i$ deviates to $p - \delta$:
        \begin{align}
            \label{ineq:by_assumption_Delta}
            & \Rev_i\left(p, \ldots, p\right) + \Delta \\
            \label{ineq:delta_to_0_holds}
            \le & \Rev_i\left(p, \ldots, p\right) + (p - \delta)\left(\Pr[B\cap A] - \frac{\Pr[A]}m\right)\\
            \label{ineq:from_prop}
            \le & \Rev_i\left(p_i = p - \delta, p_{-i} = (p, \ldots, p)\right)
        \end{align}
        Inequality~(\ref{ineq:by_assumption_Delta}) holds by the assumption we made for $\Delta$.
        Inequality~(\ref{ineq:delta_to_0_holds}) holds as $\delta \to 0^+$. (\ref{ineq:from_prop}) holds from proof~\ref{proof:defA_proposition_no_symmetric} of Proposition~\ref{proposition_no_symmetric}.
    \end{itemize}
    Therefore, in any cases for symmetric price profile $p$, it cannot be an $\epsilon$-equilibrium for $\epsilon < \Delta$.
    \end{proof}

\section{Omitted Proofs from Section~\ref{sec:buy-box}}
\proplpf*
    \begin{proof}

      \begin{enumerate}[(a)]
	\item
Consider any symmetric price profile $\mathbf p = (p, \ldots, p)$.
Compared with plain presentation, deviating to a lower price is even (weakly) more profitable, and so, 
for $p > 0$, % we show that there exists a $p' < p$ to which some seller may deviate and improve their revenue.
an argument identical to that of Proposition~\ref{proposition_no_symmetric} shows that there exists a profitable deviation $p' < p$.

For $p = 0$ and $\theta_0 > V$, we show there is a small enough $p' > 0$ such that the index $\theta'$ corresponding to $p'$ is still strictly positive.
A seller~$i$ deviating to~$p'$ sells their product when $\kappa_i > \kappa_j$ for all $j \neq i$.
The only difference from the proof of Proposition~\ref{proposition_no_symmetric} is that here seller~$j$ may be winning prominence, in which case their index is $V + 1$.
We see that essentially the same argument goes through: 
we have $\kappa_i > \kappa_j$ at least when $v_i - p' \geq \theta'$ (in which case $\kappa_i = \theta'$) and $v_j < \theta' < \max(\theta_0, V + 1)$ (in which case $\kappa_j = v_j$).
So seller $i$'s demand is at least $(1 - F(\theta' + p'))F^{n - 1}(\theta') = (1 - F(\theta_0))F^{n-1}(\theta')$.
If $\theta_0 > V$, $\theta'$ can be taken to be greater than~$V$, so that this demand is strictly positive, in which case the deviation to~$p'$ is profitable for seller~$i$.

% Let $\theta$ be the Weitzman index of a seller posting price~$p$.
% Let $A$ be the event that there are at least two sellers from whom the buyer has value at least $\theta + p$.
% For $c > 0$, $\Pr[A] > 0$.
% When $A$ happens, the index algorithm always leads the buyer to buy from such a seller.
% No matter how the tie is broken, there exists a seller~$i$ such that, the probability that the buyer from~$i$ when $A$ happens is at most $\Prx A / n$.
% On the other hand, let $B$ be the event that $v_i \geq \theta + p$. 
% Then $\Prx{B \cap A} > \Prx A / n$.
% If $i$ deviates to any $p' < p$, with a new index $\theta' > \theta$, then whenever $B$ happens, the buyer buys from~$i$.
% In particular, when both $A$ and~$B$ happens, the buyer buys from~$i$.
% When $v_i < \theta + p$, the buyer buys from seller~$i$ with a probability no less than before the deviation.
% Therefore, for any $p' < p$, the demand from seller~$i$ jump increases by at least $\Delta := \Prx{B \cap A} - \frac{\Prx{A}} n > 0$.
% Let $D$ be the demand from seller~$i$ before the deviation.
% There must exist $p' \in (0, p)$ such that $p'(D + \Delta) > pD$.
% Such a $p'$ is a profitable deviation for seller~$i$.

	\item
	This directly follows from Theorem \ref{theorem:no_BB_no_eq}. For i.i.d. sellers, the LPF mechanism prioritizes the lower-priced seller, who is also inspected first in plain presentation. The behavior under plain presentation is analogous to that under the LPF mechanism (Lemma \ref{lemma:all_buy}). Consequently, if no equilibrium exists in the plain presentation scenario, it follows that no equilibrium can be sustained under the LPF mechanism either.
      \end{enumerate}

    \end{proof}

\thmDictatorOpt*
    \begin{proof}
    %To model the buyer's decision process in a Buy Box environment effectively, 
    % We introduce two demand functions $\DBB(p, q)$ and $\DNBB(p, q)$: % which respectively denote 
    
    Consider the probability that the buyer purchases from a seller priced at $q$, when all other sellers are priced at $p$: 
    if the seller with price~$q$ is prominent, let this probability be denoted as $\DBB(p, q)$; otherwise it is denoted by $\DNBB(p, q)$.
    (Note that in the latter case, one of the sellers who post price~$p$ \emph{is} prominent.)
    It is evident that $\DBB(p, q) \geq \DNBB(p, q)$ for all $p$ and~$q$.

    Since $\mathcal{M}$ is anonymous and always allocates, if all sellers post price $p$, they win prominence each with probability $1/m$; 
    the same is true for the Dictator-$p$ mechanism.
    % the same probability at this price profile in the Dictator mechanism with dictated price $t = p$. 
    So at price profile $\mathbf p$, each seller’s revenue in $\mathcal{M}$ is the same as in the Dictator-$p$ mechanism. 
    If any seller unilaterally deviates to price $q$ in $\mathcal{M}$, with certain probability the seller is prominent, in which case her demand is $D_{\text{BB}}(q, p)$, and with the rest of the probability the seller is not in the BuyBox, while another seller with price $p$ is shown in the BuyBox, in which case the seller’s demand is $D_{\text{NBB}}(q, p)$. The same deviation to $q$ in the Dictator mechanism would result in demand $D_{\text{NBB}}(q, p)$ with probability 1. Since $D_{\text{NBB}}(q, p) \leq D_{\text{BB}}(q, p)$, if $q$ is not a profitable deviation in $\mathcal{M}$, it is not profitable in the Dictator mechanism either.
    
    % under two conditions: when the seller is displayed in the Buy Box (denoted with subscript ${}_\text{BB}$) and when they are not (${}_\text{NBB}$). 
    % These functions are critical for determining optimal pricing strategies and understanding market equilibrium.

    % README: we agreed on deleting the following part. i kept it in case we need the closed form of D_BB and D_NBB later, so that we don't need to write it again...
    % \hufu{I intend to delete all the following.  See if you agree.}
    % \note{ok. this paragraph or the whole following proof?}
    % Let $\theta_0$ be the index of a seller who posts price~0, i.e., $\theta_0$ is the solution to the equation $\mathbb E[v - \theta]^+ = c$.

    % For a seller priced at $q$ and featured in the Buy Box, the buyer buys from them with probability
    % % the demand function is given by: 
    % $$
    % \DBB(p, q) := 1 - F(\theta_0 - p + q) + \int_{v < \theta_0 - p + q} F^{m-1}(v - p + q) \: \dd F(v).
    % $$
    % % where $F$ represents the cumulative distribution function of the buyer's valuation.
    % \hufu{Give some explanation.}
    % For a seller priced at $q$ not featured in the Buy Box, % the demand function is, depending on whether $p < (>) q$:
    % $$
    % \DNBB(p, q) :=\begin{cases}
    %     \int_{V}^{\theta_0 - p + q} F^{n-1}(v) (1 - F(v - p + q)) \dd F(v) & p > q\\
    %     F(\theta_0 + p - q) \left[ 1 - F(\theta_0 - p + q) + \int_{v < \theta_0 - p + q} F^{m-2}(v - p + q) \dd F(v) \right ] & p\le q
    % \end{cases} 
    % $$
    % Observe that, for any price profile $(p, q)$, $\DNBB(p, q) \le \DBB(p, q)$, and the inequality is strict if $p \ge q - (\theta_0 - V)$.
    \end{proof}

\lemmaAllBuyBB*
    \begin{proof}
    Suppose \( \mathbf{p} = (p_1, p_2, \ldots, p_m) \) is an equilibrium price profile, ordered such that \( p_1 \geq p_2 \geq \dotsb \geq p_m \). We consider two cases based on the value of the lowest price \( p_m \):

    \begin{itemize}
        \item \textbf{Case 1: \( p_m \leq V \).}

        In the mechanism, the buyer is first presented with the prominent seller's value \( \hat{v} \) and price \( \hat{p} \) without incurring the inspection cost. The buyer's utility from the prominent seller is \( u_{\text{BB}} = \hat{v} - \hat{p} \).

        \begin{itemize}
            \item If \( u_{\text{BB}} \geq 0 \), the buyer will at least purchase from the prominent seller.

            \item If \( u_{\text{BB}} < 0 \), the buyer considers inspecting seller \( m \). The expected utility gain from inspecting seller \( m \) is:
            \begin{align*}
                & \quad \,\mathbb{E}\left[ v_m - p_m \right]^+ - c \\ 
                & = \mathbb{E}[v_m] - p_m - c\\
                & \geq \mathbb{E}[v] - V - c \geq 0.
            \end{align*}
            The above inequalities come from the assumptions \( p_m \leq V \) and \( c \leq \bar{c} = \mathbb{E}[v] - V \).

            In this case, when the lowest-price seller is not the prominent seller, if the buyer finds the prominent seller's offer to be of negative utility, the buyer finds herself better off to further inspect seller \( m \). After inspection, the buyer's utility from purchasing from seller \( m \) is at least \( v_m - p_m \geq 0 \), so the buyer will make a purchase from seller \( m \).
        \end{itemize}

        \item \textbf{Case 2: \( p_m > V \).} We will show that this case cannot occur in equilibrium.

        Consider a seller deviating to \( p_m - 1 \). If this seller is prominent, the buyer would simply purchase from this seller. But, even if the seller is not prominent, the buyer would \textbf{always} opt to inspect and purchase from this deviating seller. This is because as shown below, the expect utility gain from inspecting this deviating seller is positive, regardless of what the prominent seller offers:
        \begin{align*}
            &\quad  \mathbb{E}\left[  v - (p_m - 1)  - u_\text{BB}\right]^+ - c\\
            & \geq \mathbb{E}\left[  v - (p_m - 1)  - (V + 1 - p_m)\right]^+ - c \\
            & = \mathbb{E}\left[  v - V\right]^+ - c \geq 0.
        \end{align*}
        After inspecting this deviating seller, the buyer will prefer the deviating seller's product since the price is lower:
        \begin{align*}
            & \quad \,v - (p_m - 1) \\
            & \geq V - (p_m - 1) \\
            & = (V + 1) - p_m \\
            & \geq \max_i (v_i - p_i).
        \end{align*}
        Therefore, similar to Lemma~\ref{lemma:all_buy}, assuming \( p_m > V \), when any seller deviates and sets their price to \( p_m - 1 \), their demand increases to \( 1 \). For the equilibrium condition to hold, it is required that each seller does not find this deviation profitable; that is, for any seller \( i \in [m] \),
        \[
        \label{eq:lemma_extension_inequality_1}
            \Rev_i(\mathbf p) = p_i D_i(\mathbf{p}) \ge p_m - 1
        \]
        Notice the seller's equilibrium revenue is supposed to be strictly positive, given we assumed that $p_m - 1 \ge V - 1 \ge 0$. Summing these inequalities over \( i \), we have
        \[
        \sum_{i=1}^m p_i D_i(\mathbf{p}) \ge m (p_m - 1).
        \]
        On the other hand, any seller must have $p_i \le p_m + 1$. Because otherwise, $p_i > p_m + 1 > V + 1$ and this seller's demand and revenue would be zero. So
        $$
        \label{eq:lemma_extension_inequality_2}
            \sum_{i=1}^m p_i D_i(\mathbf{p}) \le \sum_{i=1}^m (p_m + 1) D_i(\mathbf{p}) \le p_m + 1,
        $$
        Combining the inequalities~\ref{eq:lemma_extension_inequality_1} and~\ref{eq:lemma_extension_inequality_2}, we get
        \[
        p_m \leq \frac{m - 1}{m + 1}\le 2.
        \]
        This contradicts the assumption that \( p_m > V \geq 2 \).
    \end{itemize}
    To conclude, the second case \( p_m > V \) cannot occur in equilibrium. Thus, in equilibrium, \( p_m \leq V \), and from the above analysis of Case 1, the buyer always makes a purchase at those prices.
\end{proof}

% \propLowPriceWhenCTooBig*
%     \begin{proof}
%         Consider the Dictator-$t$ mechanism. We will show that by setting $t \leq \theta_0(c) - V$ (note that $\theta_0(c) - V > 0$, because $\mathbb{E}[v - \theta_0(c)]^+ = c > \mathbb{E}[v - V]^+$), the price profile $(t, t, \ldots, t)$ will be an equilibrium.
        
%         For any seller who deviates to $p' \ne t$, $p' \geq 0$, he will lose his $\frac{1}{m}$ share of the Buy-Box. In this case, even when there are only two sellers, no buyer will inspect this deviated seller, because the Buy-Box seller's offer $u_\text{BB} = \hat{v} - t$ exceeds the inspection threshold $\theta = \theta_0(c) - p'$ for the deviated seller:
%         $$
%         \hat{v} - t > V - (V - \theta_0(c)) = \theta_0(c) > \theta_0(c) - p'.
%         $$
%         In other words, the seller receives zero revenue once he deviates from the Buy-Box threshold. Therefore, $(t, t, \ldots, t)$ is an equilibrium.
%     \end{proof}

\propHighPriceWhenCTooBig*
    \begin{proof}
    Consider the Dictator's Mechanism that sets the target price at the monopoly price:
    \[
    p^\star := \arg\max_p \, p \, (1 - F(p)).
    \]
    As long as the inspection cost is high enough ($c > \mathbb{E}[v] - \frac{p^\star}{m}$), it is optimal for every seller to maintain the monopoly price $p^\star$. This is because any seller deviating to a price $p' \neq p^\star$ would lose the prominent status and could only be inspected if all other sellers' values satisfy $(v - p^\star)_+ \leq \theta_0 - p'$. This requires that $p'$ is low enough: $p' \leq \theta_0$. Therefore, the deviating seller's revenue is at most $\theta_0$.
    
    Furthermore, as long as the monopoly profit satisfies:
    $$
    \frac{p^\star}{m} \geq \theta_0,
    $$
    and under our assumption of a high enough inspection cost $c$:
    $$
    \mathbb{E}\left[(v - \theta_0)^+\right] = c > \mathbb{E}\left[\left(v - \tfrac{p^\star}{m}\right)^+\right] \implies \frac{p^\star}{m} \geq \theta_0,
    $$
    the monopoly price vector $\vec{p}^\star = (p^\star, \dots, p^\star)$ will be an equilibrium.
    \end{proof}

\lemmaDcxExpression*
    \begin{proof}
    
    Since the equilibrium is symmetric, without loss of generality, assume that the deviating seller is seller~\(1\). We may occasionally omit the seller subscript when there is no ambiguity. Let \( x = p' - t \), where \( p' \) is the deviated price. We will analyze the demand case by cas with respect to $x$'s value.
    
    \begin{itemize}
        \item \textbf{Case \( x \leq -1 \) (i.e., \( p' < t + 1 \)).} This scenario is covered in the proof of Lemma~\ref{lemma:lemma_1_extension_for_Buy_Box}. In this case, the seller's demand is equal to \(1\).
    
        \item \textbf{Case \( -1 < x < 0 \) (i.e., \( t - 1 < p' < t \)).} Note that, in general, \( p' \) could be less than \(0\). However, this does not affect our analysis of the demand, although setting such a price would result in negative revenue for the seller.
    
        Under the Dictator Mechanism, the deviating seller is not prominent. Since all other sellers are identical, to analyze the deviating seller's demand, we assume that seller~\(2\) is prominent. There are two mutually exclusive scenarios in the buyer's search behavior that result in a purchase from seller~\(1\):
    
        \begin{itemize}
            \item \textbf{Scenario 1:} The buyer sees seller~\(2\) (prominent), inspects seller~\(1\) (not prominent), buys from seller~\(1\), and then leaves. This requires:
            \begin{enumerate}[(i)]
                \item The buyer chooses to inspect seller~\(1\) after observing seller~\(2\)'s value:
                \[
                v_2 - t < \theta_0(c) - p' \: \Leftrightarrow \ : v_2 < \theta_0(c) - x.
                \]
                \item After inspecting seller~\(1\), the buyer prefers seller~\(1\) over seller~\(2\):
                \[
                \label{condition_2.1}
                v_1 - p' \geq v_2 - t \: \Leftrightarrow \: v_1 \geq v_2 + x.
                \]
                \item After inspecting seller~\(1\), the buyer does not inspect any additional sellers:
                \begin{align*}
                    & \max(v_1 - p', v_2 - t) \geq \theta_0(c) - t  \\
                    \Leftrightarrow\: & \underbrace{\max(v_1 - x, v_2) = v_1 - x}_{\text{by condition~\ref{condition_2.1}}} \geq \theta_0(c).
                \end{align*}
            \end{enumerate}
            The demand contribution from this scenario is
            \[
            \int_V^{\max(\theta_0(c) - x, V)} \int_{\max(\theta_0(c) + x, v_2 + x)}^{V + 1} \dd F(v_1)\, \dd F(v_2).
            \]
    
            \item \textbf{Scenario 2:} The buyer first inspects seller~\(2\), then inspects all sellers, compares all products, and finally buys from seller~\(1\). This requires:
            \begin{enumerate}[(i)]
                \item The buyer chooses to inspect seller~\(1\) after observing seller~\(2\)'s value:
                \[
                v_2 - t < \theta_0(c) - p' \: \Leftrightarrow \: v_2 < \theta_0(c) - x.
                \]
                \item After inspecting all sellers, the buyer prefers seller~\(1\)'s product:
                \label{condition_2.2}
                \begin{align*}
                & v_1 - p' \geq \max_{i = 2, \ldots, m} \{ v_i - t \} \\
                \Leftrightarrow \: &  v_1 - x \geq \max_{i = 2, \ldots, m} \{ v_i \}.
                \end{align*}
                \item The buyer chooses to inspect all remaining sellers (\( i = 3, \ldots, m \)):
                \begin{align*}
                    \max\left(v_1 - p', \max_{i = 2, \ldots, m} \{ v_i - t \}\right) \geq \theta_0(c) - t \\
                    \Leftrightarrow \underbrace{\max\left(v_1 - x, \max_{i = 2, \ldots, m} \{ v_i \}\right) = v_1 - x}_{\text{by condition~\ref{condition_2.2}}} \geq \theta_0(c).
                \end{align*}
            \end{enumerate}
            The demand contribution from this scenario is
            % \[
            % \int_V^{\max(\theta_0(c) - x, V)} \int_{v_2 + x}^{\max(\theta_0(c) + x, v_2 + x)} F^{m-2}(v_1 - x) \, \dd F(v_1)\, \dd F(v_2).
            % \]
            \begin{align*}
                \int_V^{\max(\theta_0(c) - x, V)} & \int_{v_2 + x}^{\max(\theta_0(c) + x, v_2 + x)} \\
                & F^{m-2}(v_1 - x) \, \dd F(v_1)\, \dd F(v_2).
            \end{align*}
        \end{itemize}
    
        These two scenarios are mutually exclusive. Therefore, the total demand of the deviating seller when setting the price \( p' = t + x \) for \( x \in (-1, 0) \) is
        % \begin{align*}
        %     \D_c(x) &:= \int_V^{\max(\theta_0(c) - x, V)} \int_{\max(\theta_0(c) + x, v_2 + x)}^{V + 1} \dd F(v_1)\, \dd F(v_2)\\
        %     & \quad + \int_V^{\max(\theta_0(c) - x, V)} \int_{v_2 + x}^{\max(\theta_0(c) + x, v_2 + x)} F^{m-2}(v_1 - x) \, \dd F(v_1)\, \dd F(v_2).
        % \end{align*}
        \begin{align*}
            & \quad \D_c(x) :=  \\
            &\int_V^{\max(\theta_0(c) - x, V)} \int_{\max(\theta_0(c) + x, v_2 + x)}^{V + 1} \dd F(v_1)\, \dd F(v_2) \\
            &\quad + \int_V^{\max(\theta_0(c) - x, V)} \int_{v_2 + x}^{\max(\theta_0(c) + x, v_2 + x)} \\
            &\quad\quad F^{m-2}(v_1 - x) \, \dd F(v_1)\, \dd F(v_2).
        \end{align*}
        \item \textbf{Case \( 0 \leq x < \theta_0(c) - V \) (i.e., \( t \leq p' < \theta_0(c) \)).} In this case, particularly for \( p' = t \) and define \( \D_c(0) := \lim_{x \to 0^+} \D_c(x) \). For \( x \in (0, \theta_0(c) - V) \), since the deviating seller's price \( p' = t + x > t \), they are inspected last only if the values of all other sellers \( i = 2, \ldots, m \) satisfy \( v_i \leq \theta_0(c) - x \). Additionally, the deviating seller~\(1\) is purchased from only if
        \[
        v_1 - p' \geq \max_{i = 2, \ldots, m} \{ v_i - t \}.
        \]
        These conditions translate to
        \[
        \D_c(x) = \int_V^{\theta_0(c) - x} \left(1 - F(v + x)\right) \, \dd F^{m - 1}(v).
        \]
    
        \item \textbf{Case \( x \geq \theta_0(c) - V \) (i.e., \( p' > \theta_0(c) \)).} In this scenario, the deviating seller~\(1\) will never be inspected. Suppose that seller~\(2\) is prominent and yields a value \( v_2 \):
        \begin{align*}
            v_2 - t & \geq V - t \\
            & = V - (p' - x) = V - x - p' \geq \theta_0(c) - p'.
        \end{align*}
        Since buyers will never choose to inspect seller~\(1\), the seller's demand is zero.
    \end{itemize}
    
    This concludes the proof. The deviating seller's demand is thus expressed as a function of the deviation \( x = p' - t \), \emph{independent from $t$}, where \( p' \) is the deviated price and \( t \) is the symmetric equilibrium price of the remaining sellers.
    \end{proof}

\thmTc*
    \begin{proof}
    Theorem \ref{theorem:dictator_opt} states that it is sufficient to consider dictator's mechanism for any implementable equilibrium price. Consider Dictator mechanism setting target price at $t$, inducing symmetric equilibrium $(t, t, \ldots, t)$. 
        
    If any seller deviates to $p' = t + x$, Lemma~\ref{lemma:D_c_independent_of_t} gives the seller's post-deviate demand expressed in $x$: $\D_c(\cdot)$. The equilibrium condition can then be expressed as:
    \begin{align}
    \label{key_EQ_condition}
        \underbrace{(x + t)}_{p' = x + t} \cdot \underbrace{\D_c(x)}_{\text{seller $1$'s demand}}  \le \underbrace{\frac1m t}_{\text{Eq. revenue}}, & \quad \forall x \in \mathbb R
    \end{align}
    By transforming this condition algebraically, we obtain
    \begin{align}
    \label{key_EQ_condition_variant}
    \left(\frac 1m - \D_c(x)\right)t \ge x\D_c(x).
    \end{align}
    
    Discuss over potential value of $x$ (i.e. $p' = t + x$).
    \begin{itemize}
        \item \textbf{Case $x > 0$ (i.e. $p = x + t > t$}). It's evident that post-deviate demand $\D_c(x) < \frac 1m$. Hence the above condition translates to
        $$
        t \ge \frac{x\D_c(x)}{\frac 1m - \D_c(x)}, \forall x > 0.
        $$
        This gives a lower bound for implementable price $t$.
        \item \textbf{Case $x < 0$ (i.e. $p = x + t < t$}). Notice that the left-hand-side of the inequality~\ref{key_EQ_condition_variant}, $x\D_c(x) < 0$ --- so the condition automatically holds for $\frac1m - \D_c(x) \ge 0$. Now considers $x$ such that $\D_c(x) > \frac1m$, inequality~\ref{key_EQ_condition_variant} translates to
        $$
        t \le \frac{(-x)\D_c(x)}{\D_c(x) - \frac1m},\quad  \forall x < 0 \text{ and }\frac1m - \D_c(x) < 0
        $$.
    \end{itemize}
    And this sums up to be the condition for the theorem.
\end{proof}

 \exampleTc*
 \begin{proof}[Find $t^*(c), \bar t(c)$ for two uniform seller.] 

 For convenience, use $\theta$ to denote $\theta_0(c)$.
 First calculate $\D_c(x)$:
 \begin{align*}
     \D_c(x) = \begin{cases}
        1 & x < - 1\\
        \frac12 x^2  + \frac12 & x \in [-1, \theta - 3)\\
        -x + \frac12(\theta - 2)^2 & x\in [\theta - 3, 0)\\
        -\frac12 x^2 + \frac12 (\theta - 2)^2 & x \in [0, \theta - 2)\\
        0 & x\ge \theta - 2
     \end{cases}
 \end{align*}
Find $t^*(c)$ and $\bar t(c)$: let
$$
\chi^c(x) := \frac{x \D_c(x)}{\frac12 - \D_c(x)}
$$
then
$$
\chi^c{}'(x) = \frac1{(\frac12 - \D_c(x))^2}\left(\D_c(x)(\frac12 - \D_c(x)) + \frac12 x D_c'(x)\right).
$$
Let $\tilde \chi(x) := \D_c(x)(\frac12 - \D_c(x)) + \frac12 x D_c'(x)$,
$$
\tilde \chi^c{}'(x) = \D'_c(x)(1 - 2\D_c(x)) + \frac12 x \D_c''(x).
$$
To get $t^*(c) = \sup_{x > 0}\chi^c(x)$, notice that only needs to consider $x\in (0, \theta - 2)$ where $\D_c(x) > 0$. Within this region, $\D_c(x) < \frac12$, $\D_c(x) = -x < 0$ and $\D_c''(x) = -1$, therefore
$$
\tilde \chi^c{}'(c) \le 0.
$$
And
\begin{small}
    $$
    \tilde \chi^c(x) =  - \frac14 x^4 + (\frac12(\theta - 2)^2 - \frac12)x^2 + \frac12(\theta - 2)^2 \left(\frac12 - \frac12(\theta - 2)^2\right).
    $$
\end{small}
So we can solve for the unique interior supremum point $x^*$ that maximizes $\chi^c(x)$:
$$
(x^*_c)^2= -(\hat \theta + \frac12) + \sqrt{2\hat \theta + \frac14},
$$
where $\hat\theta = 1- (\theta - 2)^2$.
Plug in $x^*_c$ into $\chi^c(x)$ (because $t^*(c) = \sup_{x > 0}\chi^c(x)$) we obtain 
\begin{align*}
    t^*(c) & = \frac{x^*_c \D_c(x^*_c)}{\frac12 - \D_c(x^*_c)}\\
    & =  \frac{\sqrt{-(\hat \theta + \frac12) + \sqrt{2\hat \theta + \frac14}}\times \frac12\left(\frac32 - \sqrt{2\hat \theta + \frac14}\right)}{\frac12\sqrt{2\hat \theta + \frac14} - 1}\\
    & = \frac{\sqrt{-((2\sqrt{2c} - 2c) + \frac12) + \sqrt{2(2\sqrt{2c} - 2c) + \frac14}}}{\frac12\sqrt{2(2\sqrt{2c} - 2c) + \frac14} - 1}\\
    & \quad \times \frac12\left(\frac32 - \sqrt{2(2\sqrt{2c} - 2c) + \frac14}\right)
\end{align*}
Vice versa, we can directly calculate $\bar t(c)$.

For $x\in [\theta - 3, 0)$, let $\tilde \theta = (\theta - 2)^2$, so $\D_c(x) = -x + \frac12 \tilde \theta$, and
\begin{small}
\begin{align*}
\chi^c(x) & = \frac{x(-x + \frac12 \tilde \theta)}{\frac12 - ( -x + \frac12 \tilde \theta)}\\
& = (\frac12 \tilde \theta -x - \frac12) + \frac{\frac14(1 - \tilde \theta)}{(\frac12 \tilde \theta -x - \frac12)} + 1 - \frac12 \tilde \theta\\
& \ge 2.
\end{align*}
\end{small}
As for $x\in [-1, \theta - 3)$,
\begin{align*}
    \chi^c(x) & = \frac{x(-x^2 + \frac12)}{-\frac12 x^2}\\
    & = -(x + \frac1x) \ge 2.
\end{align*}
Therefore,
$$
\bar t(c) \equiv 2.
$$
\end{proof}

\propTStarProperties*
    \begin{proof}[Proof of Proposition~\ref{prop:t_star_properties}]
    We will first prove that $t^*(\cdot)$ is monotonically decreasing and continuous. Then prove its endpoint property.
    \paragraph{Monotonicity of $t^*(c)$.}
    Theorem~\ref{theorem:m_seller_general_T(c)} gives
    $$t^*(c) := \sup_{x > 0} \left\{\frac{x \D_c(x)}{\frac1m - \D_c(x)} \right\}.$$
    Fix $ x $, consider inspection costs $ c_1 \leq c_2 $ with $\theta_0(c_1) \geq \theta_0(c_2)$, we have, for $x \ge 0$:    
    \begin{align*}
        \D_{c_1}(x) - \D_{c_2}(x) 
        & = \int_V^{\theta_0(c_1) - x} (1 - F(v + x))\, \dd F^{m-1}(v) \\
        & \quad - \int_V^{\theta_0(c_2) - x} (1 - F(v + x))\, \dd F^{m-1}(v) \\
        & = \int_{\theta_0(c_2) - x}^{\theta_0(c_1) - x} (1 - F(v + x))\, \dd F^{m-1}(v) \geq 0.
    \end{align*}
    So, fix $x$, $\D_c(x)$ decreases as $c$ increases. Therefore, for any $c_1 \ge c_2$, $\D_{c_1}(x) \le \D_{c_2}(x)$. For any $x > 0$:
    \begin{align*}
        \frac{x\D_{c_1}(x)}{\frac1m - \D_{c_1}(x)} 
        & = (-x) + \frac{\frac1m}{\frac1m - \mathcal D_{c_1}(x)} \\
        & \le (-x) + \frac{\frac1m}{\frac1m - \mathcal D_{c_2}(x)} = \frac{x\D_{c_2}(x)}{\frac1m - \D_{c_2}(x)}.
    \end{align*}

    This implies $t^*(c_2) \ge t^*(c_1)$, hence proving $t^*(\cdot)$ monotonically decreases. 

    \paragraph{Continuity of $t^*(c)$.} 
    The expression $t^*(c) := \sup_{x > 0} \left\{\frac{x \D_c(x)}{\frac1m - \D_c(x)} \right\}$ takes supremum over $x > 0$. Because $\D_c(x) = 0$ for $x\ge \theta_0(c) - V$, it suffices to consider the open interval $x\in (0, \theta_0(c) - V)$. 
    
    Define $\phi(c, x) := \frac{x\D_c(x)}{\frac{1}{m} - \D_c(x)}$. The function $\phi$ is jointly continuous in $(c, x)$ for $x > 0$. Given that the supremum point $x^*$ that maximizes $\phi(c, x)$ within the interval $(0, \theta_0(c) - V)$ is obtained in an interior point, by pointwise convergence, $t^*(c) = \inf_{x \in (0, \theta_0(c) - V)} \phi(c, x)$ will also exhibit continuity. 
    
    \paragraph{The limiting convergence as $c \to \bar c^-$.}
    Discuss over the region of $x$ that we're taking supremum over for $t^*(c) := \sup_{x > 0} \left\{\frac{x \D_c(x)}{\frac1m - \D_c(x)} \right\}$:
    \begin{itemize}
        \item $x\in [\theta_0(c) - V, \infty)$, $\D_c(x) = 0$. So this part doesn't contribute to $t^*(c)$.
        \item $x \in (0, \theta_0(c) - V)$:
        Look at $\D_c(x)$ as $x\to 0^+$:
        \begin{align*}
            \lim_{x\to 0^+} \D_c(x) & = \lim_{x\to 0^+} \int_V^{\theta_0(c) - x} \left(1 - F(v + x)\right) \, \dd F^{m-1}(v)\\
            & = \int_V^{\theta_0(c)} \left(1 - F(v)\right) \, \dd F^{m-1}(v)\\
            & = \int_V^{\theta_0(c)} \dd F^{m-1}(v) - \frac{m-1}m \int_V^{\theta_0(c)} \dd F^m(v)\\
            & = F^{m-1}(\theta_0(c)) - \frac{m-1}m F^m(\theta_0(c)),
        \end{align*}
        we have, as $x\to 0^+$
        \begin{align*}
            & \quad  \lim_{x\to 0^+}\left(\frac1m - \D_c(x)\right) \\
            & = \frac1m - \left(F^{m-1}(\theta_0(c)) - \frac{m-1}m F^m(\theta_0(c))\right).
        \end{align*}
        Let $g(\alpha) = \frac1m - \left(\alpha^{m-1} - \frac{m-1}m \alpha^m \right)$, for $\alpha\in [0, 1]$. $g(1) = 0$, $g'(\alpha) = -(m - 1)\alpha^{m-2}(1 - \alpha) < 0, \forall \alpha\in (0, 1)$. So $g(\alpha) > 0, \forall \alpha \in (0, 1)$. As $c \to \bar c^+$, $\theta_0(c) \to V$ so for sure $F(\theta_0(c)) \in(0, 1)$ (actually $F(\theta_0(c))$ is bounded far away from $1$). Plug in $\alpha = F(\theta_0(c))$, we have
        $$
        \lim_{x\to 0^+}\left(\frac1m - \D_c(x)\right) = g(F(\theta_0(c))) > 0.
        $$
        Notice that as $c\to c^-$, $\theta_0(c) \to V^+$:
        \begin{align*}
            & \quad \lim_{c\to \bar c^-} \sup_{x\in (0, \theta_0(c) - V)} \left\{\frac{x \D_c(x)}{\frac1m - \D_c(x)} \right\} 
            \\
            & \le \lim_{c\to \bar c^-} \left(\theta_0(c) - V\right) \sup_{x\in (0, \theta_0(c) - V)} \left\{\frac{\D_c(x)}{\frac1m - \D_c(x)} \right\} \\
            & \le \lim_{c\to \bar c^-} \left(\theta_0(c) - V\right) \sup_{x\in (0, \theta_0(c) - V)} \left\{\frac{1}{\frac1m - \D_c(x)} \right\}
            \\
            & = 0.
        \end{align*}
    \end{itemize}
    Taking together the two regions of $x$ discussed above, $\lim_{c\to \bar c^-}  t^*(c) = 0$.
    \paragraph{The limiting convergence as $c \to 0^+$.}
    First, consider the equilibrium \textbf{without} inspection cost. In this case, the demand function for the deviated seller $\D_c(x)$ naturally extends to $c = 0$, when $\theta_0(0) = V + 1$:
        $$
        \D_0(x) = \int_V^{V + 1 - x} (1 - F(v + x))\, \dd F^{m-1}(v).
        $$
        Denote the equilibrium price as $(t_0, t_0, \ldots, t_0)$. The equilibrium condition for any deviate $x\in [\max(-1, -t_0), 1]$:
        \begin{align}
                & \quad\: (t_0 + x) \D_0(x) \le \frac 1m t\\
                \label{eq_condition:c=0}
                & \Leftrightarrow \D_0(x) \le \frac1m\frac{t_0}{t_0 + x}, \forall x\in [\max(-1, -t_0), 1)
        \end{align}
        Notice that, the inequality is strict at $x = 0$---when $\D_0(x) = \frac1m  = \frac1m\frac{t_0}{t_0 + x}\big |_{x = 0}$. And since $\D_0(x)$ is continuous, this situation implies that the slope of $\D_0$ and the right-hand side's equal-revenue curve coincide, yielding:
        $$
        t_0 = \frac{1}{m(m-1)\int_V^{V + 1}F^{m-2}(v)f^2(v)\, \dd v}.
        $$
        If (\ref{eq_condition:c=0}) holds for $t_0$, then $(t_0, \ldots, t_0)$ constitutes the unique symmetric equilibrium when the market is free of inspection costs. 
        
        As for the limiting equilibrium price when inspection cost approaches zero, $t^*(0)$, because $t^*(c) = \inf_x \phi(c, x)$ is continuous in $c$, all that remains is to verify $\sup_x \phi(0, x) = t_0$. This is confirmed at the limit as $x \to 0^+$:
        \begin{align*}
            \lim_{x \to 0^+}\phi(0, x) & = \lim_{x \to 0^+}\frac{x\D_0(x)}{\frac{1}{m} - \D_0(x)} \\
            & = \frac{\frac{\partial}{\partial x}x\D_0(x)}{\frac{\partial}{\partial x}[\frac{1}{m} - \D_0(x)]}\bigg|_{x = 0}\\
            & = \frac{1}{m(m-1)\int_V^{V + 1}F^{m-2}(v) f(v)^2\, \dd v}.
        \end{align*}
        By the monotonicity and continuity of $t^*(c)$, it would be contradictory for some $x > 0$ if $\phi(0, x) > t_0$, as this would imply, for some negligible cost of inspection $\epsilon \to 0$, that $\phi(\epsilon, x) > t_0$, leading to a contradiction. Therefore, presuming $t_0$ exists as a symmetric equilibrium absent inspection costs, $\lim_{c \to 0^+}t^*(c) = t_0$.
     \end{proof}

\propTBarProperties*
    \begin{proof}
        Theorem~\ref{theorem:m_seller_general_T(c)} states that
        $$
        \bar t(c) := \inf_{x < 0} \left\{\left(\frac{(-x) \D_c(x)}{\D_c(x) - \frac1m}\right)^+\right\}\setminus \{0\},
        $$
        so we focus on $x < 0$ and study the behavior of $\D_c(x)$ with respect to $c$. For convenience, denote
        \begin{small}
            \begin{align*}
                \tilde h_x(c)  & = \D_c(x) \\
                & = \int_V^{\theta_0(c) -x} \int_{\max(\theta_0(c) + x, v_2 + x)}^{V + 1} \, \dd F(v_1)\, \dd F(v_2) \\
                & \quad + \int_V^{\theta_0(c) -x} \int_{v_2 + x}^{\max(\theta_0(c) + x, v_2 + x)} F^{m-2}(v_1 - x)\, \dd F(v_1)\, \dd F(v_2).
            \end{align*}
        \end{small}
    
    Note that this is still inconvenient. Let
    \begin{small}
    \begin{align*}
        h_x(\theta) & := \tilde h_x(c) \circ \theta_0(c) \\
        & = \int_V^{\theta -x} \int_{\max(\theta + x, v_2 + x)}^{V + 1} \, \dd F(v_1)\, \dd F(v_2) \\
        & \quad + \int_V^{\theta -x} \int_{v_2 + x}^{\max(\theta + x, v_2 + x)} F^{m-2}(v_1 - x)\, \dd F(v_1)\, \dd F(v_2)\\
        & = \int_V^\theta \left(1 - F(\theta + x)\right)\, \dd F(v_2) \\
        & \quad + \int_V^\theta \int_{v_2 + x}^{\theta + x} F^{m-2}(v_1 - x)\, \dd F(v_1) \, \dd F(v_2)\\
        & \quad + \int_{\theta}^{\theta - x} \int_{v_2 + x}^{V + 1}\, \dd F(v_1)\, \dd F(v_2)\\
        & = \left(1 - F(\theta + x)\right) F(\theta) \\
        & \quad + \int_V^\theta \int_{v_2 + x}^{\theta + x} F^{m-2}(v_1 - x)\, \dd F(v_1) \, \dd F(v_2)\\
        & \quad + \int_\theta^{\theta - x} \left(1 - F(v_2 + x)\right)\, dF(v_2)
    \end{align*}
    \end{small}
    So
    \begin{align*}
        h_x'(\theta) & = -f(\theta + x) F(\theta) + \int_V^\theta F^{m-2}(\theta)f(\theta + x)\, \dd F(v_2)\\
        & \quad + \left(1 - F(\theta)\right) f(\theta - x) - \left(1 - F(\theta + x)\right) f(\theta)\\
        & = \left(F^{m-2}(\theta) - 1\right) f(\theta + x)F(\theta) + \left(1 - F(\theta)\right) f(\theta - x).
    \end{align*}
    Since $\theta_0(c)$ is the solution for $\mathbb E[v - \theta]^+ = c$, $\theta_0'(c) \le 0$, and this implies
    $$
    \tilde h_x'(c) = \theta_0'(c) h_x'(\theta).
    $$
    When $m = 2$
    $$
    h_x'(\theta) = \left(1 - F(\theta)\right) f(\theta - x) \ge 0.
    $$
    This implies that $\tilde h_x'(c)\le 0$, and that the demand for deviated seller with $p' > t$ would decrease as inspection cost $c$ increases.
    Thus, $ c_1 \leq c_2 \Rightarrow \D_{c_1}(x)\ge  \D_{c_2}(x), \forall x < 0$ when $m = 2$. So $ \forall x < 0$,
    \begin{align}
        \frac{(-x) \D_{c_1}(x)}{\D_{c_1}(x) - \frac1m} & = (-x) + \frac{(-x)\frac1m}{\D_{c_1}(x) - \frac1m} \\
        & \le (-x) + \frac{(-x)\frac1m}{\D_{c_2}(x) - \frac1m} = \frac{(-x) \D_{c_2}(x)}{\D_{c_2}(x) - \frac1m}.
    \end{align}
    And this implies $\bar t(c_1) \le \bar t(c_2), \forall 0 < c_1 \le c_2 \le \bar c$, thereby confirming the monotone property of $ \bar t(c) $.
    \end{proof}

% \corrNested*

% \note{proof pending... it's too straightforward that i don't quite know where to start.}

\thmThresholdTc*
    \begin{proof}
        We investigate the robustness of the Threshold-$t$ mechanism in sustaining an equilibrium.
        Assume that symmetric equilibrium $(t, t, \ldots, t)$ holds, where $t$ is the threshold below which a seller would win prominence uniformly at random. Consider seller $1$ who deviates from equilibrium price $t$ to $p_1' = t + x$. When $x > 0$, the buyer will find seller $1$'s offer least attractive and only inspected him at last if all other sellers are not satisfying enough (exactly similar to the case in Dictator's mechanism)
        \begin{align*}
            D_1^{\text{Threshold}}(p_1', t) = \D_c(p_1' - t).
        \end{align*}
        When seller $1$ deviates to a lower price, i.e. $x < 0$, he wins prominence with probability $1/m$. For the rest $1 - 1/m$, he is not prominent, conditional on which his demand corresponds to the case as if he is in Dictator mechanism. We can derive seller $1$'s demand as follows:
        \begin{align*}
            D_1^{\text{Threshold}}(p_1', t) 
            & = \frac1m D_{BB}(p_1', t) + (1 - \frac1m)\D_c(x)
        \end{align*}
        where, $D_{BB}(p_1', t)$ is seller $1$'s demand when he is prominent:
        \begin{small}
            \begin{align*}
                D_{BB}(p_1', t) & = P[\text{buyer buys seller $1$ without inspecting others}] \\
                & \quad + P[\text{buyer inspect all sellers and buys $1$}]\\
                & = P[v_1 - p_1' + t \ge \theta] \\
                & \quad + P[v_1 - p_1' + t \ge \theta \text{ and } v_1 - p_1' \ge \max_{i = 2, \ldots, m} v_i - t]\\
                & = P[v_1 - x \ge \theta] + P[v_1 - x \ge \theta \text{ and } v_1 - x \ge \max_{i = 2, \ldots, m} v_i]\\
                & = 1 - F(\theta + x) + \int_V^{\theta + x} (1 - F(v + x))\,\dd F^{m-1}(v).
            \end{align*}
        \end{small}
        Notice that, the demand of the deviated seller can again be delineated w.r.t. the price differential $x$. Denote $\tilde \D_c(x) := D_1^{\text{Threshold}}(p, t)$ as the demand of seller $1$. Specifically, for $x = 0$, let $\tilde \D_c(x) = \frac1m t$ so that it's left-continuous. And $\tilde \D_c(x)$ can be summarized as:
        % \begin{align*}
        %     \tilde \D_c(x) & = \begin{cases}
        %         \frac1m \D_c(x) + (1 - \frac1m)\left(1 - F(\theta_0(c) + x) + \int_V^{\theta_0(c) - x} (1 - F(v + x))\,\dd F^{m-1}(v)\right) & x < 0\\
        %         \D_c(x) = \int_{V}^{\theta - V} (1 - F(v + x))\, \dd F^{m-1}(v) & x > 0.
        %     \end{cases}
        % \end{align*}
    \begin{small}
        \begin{align*}
            \tilde \D_c(x) & = \begin{cases}
                \frac{1}{m} \D_c(x) + (1 - \frac{1}{m})\times \\
                \quad \left(1 - F(\theta_0(c) + x) + \int_V^{\theta_0(c) - x} (1 - F(v + x))\,\dd F^{m-1}(v)\right) \\
                \quad \quad \quad \quad \quad \quad \quad \quad \quad \quad \quad \quad  \text{for } x < 0 \\
                \D_c(x) = \int_{V}^{\theta - V} (1 - F(v + x))\, \dd F^{m-1}(v) \\
                \quad \quad \quad \quad \quad \quad 
                \quad \quad \quad \quad \quad \quad \text{for } x > 0.
            \end{cases}
        \end{align*}
    \end{small}

        Threshold-$t$ mechanism can sustain a symmetric equilibrium $(t,t, \ldots, t)$ if and only if $t \in \tilde T(c)$, for $\tilde T(c)$ being articulated as:
        \begin{align*}
            \tilde T(c) := \left\{t : (x + t)\tilde \D_c(x) \le \frac{1}{m} t, \forall x \in [\max(-1, -t), 1]\right\}.
        \end{align*}

        Following the methodological approach in the proof of Theorem \ref{theorem:m_seller_general_T(c)}, $\tilde T(c)$, for $t\in \tilde (c)$ should satisfies the following criteria simultaneously:
        \begin{align*}
        t & \le \inf_{x < 0}\frac{(-x)\tilde \D_c(x)}{\tilde \D_c(x) - \frac{1}{m}},\\
        t & \ge \sup_{x > 0}\frac{x\tilde \D_c(x)}{\frac{1}{m} - \tilde \D_c(x)}.
        \end{align*}
        The set $\tilde T(c)$ is nonempty, indicating the existence of a symmetric equilibrium $(t, \ldots, t)$, if and only if:
        $$
        \sup_{x > 0}\frac{x\tilde \D_c(x)}{\frac{1}{m} - \tilde \D_c(x)} \le \inf_{x < 0}\frac{(-x)\tilde \D_c(x)}{\tilde \D_c(x) - \frac{1}{m}}.
        $$
        Let $\hat t(c) := \inf_{x < 0}\frac{(-x)\tilde \D_c(x)}{\tilde \D_c(x) - \frac{1}{m}}, t^*(c) := \sup_{x > 0}\frac{x\tilde \D_c(x)}{\frac{1}{m} - \tilde \D_c(x)}$. For the Threshold-$t$ mechanism with inspection cost $c$, $t$ can sustain an equilibrium if and only if $t^*(c) \le \hat t(c)$. If so, the implementable region of $t$ for Threshold mechanism is given by the following interval
        \begin{align*}
            [t^*(c), \hat t(c)].
        \end{align*}
    \end{proof}

\corrThresholdSucks*
    \begin{proof}
        We examin a specific instance of two i.i.d. sellers, where their valuations are distributed according to the cumulative distribution function $F(v) = \frac{1}{e-1}(e^{v - 2} - 1)$, supported on $[2, 3]$. 
        By Theorem~\ref{theorem:dictator_opt}, it suffices to look at Dictator's mechanism for all symmetric equilibrium supported by the class of \emph{standard} mechanism.
        Therefore, for Dictator's mechanism and Threshold mechanism, the implementable price regions under symmetric equilibrium is given by Theorem~\ref{theorem:threshold_T(c)} and Theorem~\ref{theorem:m_seller_general_T(c)}:
        \begin{align*}
            \text{Dictator }T(c) & = [t^*(c), \bar t(c)]\\
            \text{Threshold }\tilde T(c) & = [t^*(c), \hat t(c)]\\
        \end{align*}
        We vary inspection cost $c$ to investigate the change of implementable regions, as shown in Figure~\ref{fig:dictator_threshold_separation}. It can be observed that, the upper bound $\hat t(c)$ of Threshold mechanism is lower than that of Dictator's mechanism. This implies, Threshold mechanism cannot sustain any symmetric equilibrium price above that line, yet Dictator's mechanism goes far beyond that bound and can sustain even higher prices in equilibrium.
        % \begin{figure}
        %     \centering
        %     \includegraphics[width=\linewidth]{figures/DvT_simulation.png}
        %     \caption{This figure shows the difference in sustaining symmetric equilibrium, between Dictator and Threshold mechanism. We vary inspection cost $c$ ($x$-axis) and color corresponding implementable regions $T(c), \tilde T(c)$ for the two mechanism.}
        %     \Description{Illustrative example of the difference between Dictator and Threshold mechanisms under various inspection costs. \newline For each value of inspection cost $c\in (0, \bar c)$, we calculate and plot the upperbound $\bar t(c)$[$\tilde t_1(c)$] for Dictator[Threshold] mechanisms, and their mutual lowerbound $t^*(c)$. We visualize the feasible $t$-region under different inspection cost values.}
        % \label{fig:dictator_threshold_separation}
        % \end{figure}
    \end{proof}

\section{Omitted Proofs from Section~\ref{sec:SW_CS_analysis}}

\SWDeterminedByCnMonotone*

\begin{proof}
Consider any symmetric pure equilibrium of some standard Buy Box mechanism, where the price is~$t$.
We show that the welfare does not depend on~$t$.
By Lemma~\ref{lemma:lemma_1_extension_for_Buy_Box}, we have $t \leq V$, i.e., the buyer always makes a purchase.

Without loss of generality, assume seller~$1$ is in the Buy Box.
By Proposition~\ref{prop:define_theta_0}, the index of each seller $i$ not in the Buy Box is $\theta_i(t) = \theta_0(c) - t$, and the index of seller~1 is $\theta_1(t) = V + 1 - t$.
Recall from Section~\ref{sec:prelim} that the buyer's surplus is $\Ex{\max_i \kappa_i(t)}^+$, where $\kappa_i(t) = \min(\theta_i(t), v_i - t)$.
Since $\theta_0 \geq V$ by assumption (as $c < \bar c$), and since $t \leq V$, we see that $\kappa_i(t) \geq 0$ for each $i$. 
For any other equilibrium price $t'$, we have $\theta_i(t') - \theta(t) = t - t'$ for each~$i$, and hence $\kappa_i(t') - \kappa_i(t') = t - t'$.
Therefore, $\Ex{\max_i \kappa_i(t')}^+ - \Ex{\max_i \kappa_i(t)}^+ = t - t'$.
That is to say, the consumer surplus changes linearly in the opposite direction with the equilibrium price.

On the other hand, the total revenue of the sellers is exactly the equilibrium price since the buyer always makes a purchase.  
Therefore the welfare does not depend on the equilibrium price.

The fact that the welfare decreases monotonically as $c$ increases is immediate from the argument above, by noticing that $\theta_0(c)$ decreases with~$c$.
\end{proof}

\thSWclosedform*
    
    \begin{proof}
            We first solve in closed form for social welfare, and demonstrate how it might be interpreted as the simplified representation stated in the theorem. 
            
            Before we start we simplify some notations. Because every $c \in [0, \bar c]$ corresponds to a unique index value $\theta_0\in [V, V + 1]$, we replace $c$ and $\theta_0(c)$ with shortened expression $\theta$ when there is no confusion. Also, in some expressions, replace inspection cost $c$, $\theta_0(c)$ with corresponding index $\theta$. For example, social welfare at a specific inspection cost level ($\SW(c)$) as $\SW(\theta)$. For two i.i.d. sellers, the match value $\nu_c$ (or, denoted as $\nu_\theta$, equals
            $$\nu_\theta= P[v_2 \ge \theta]E[v_2|v_2 \ge \theta] + P[v_2 < \theta]E[\max(v_1, v_2) | v_2 < \theta].$$
            Social welfare can be expressed as (WLOG assumes seller $2$ is inspected first),
            \begin{small}
                \begin{align*}
                    SW(\theta) & = -P[v_2 \le \theta] c + u_\theta\\
                    & = -P[v_2\le \theta]c \\
                    & \quad + P[v_2 \ge \theta]E[v_2|v_2 \ge \theta]\\
                    & \quad + P[v_2 < \theta]E[\max(v_1, v_2) | v_2 < \theta]\\
                    & = -cF(\theta) \\
                    & \quad + \int_\theta^{V+1} v_2\dd Fv_2 \\
                    & \quad + \int_V^{\theta}(\int_{V}^{\theta} \max(v_1, v_2) \ \dd Fv_1 + \int_\theta^{V+1}v_1 \dd Fv_1) \dd Fv_2\\
                    & = -cF(\theta) \\
                    & \quad + (\int_\theta^{V+1} v\ \dd Fv)(1 + F(\theta)) \\
                    & \quad + \int_V^\theta (\int_V^{v_2} v_2\ \dd F(v_1)  + \int_{v_2}^\theta v_1\ \dd F(v_1))\ \dd F(v_2)\\
                    & = -cF(\theta)\\
                    & \quad + (\int_\theta^{V+1} v\ \dd Fv)(1 + F(\theta)) \\
                    & \quad + \int_V^\theta v_2 F(v_2)\ \dd F(v_2) + \int_V^\theta \int_{v_2}^\theta v_1 \ \dd F(v_1)\dd F(v_2)
                \end{align*}
            \end{small}        
        Whereas the last component above
        $
            \int_V^\theta \int_{v_2}^\theta v_1 \ \dd F(v_1)\ \dd F(v_2) = \int_V^\theta \int_V^{v_1} v_1\ \dd F(v_2) \dd F(v_1) =\int_V^\theta v_1 F(v_1) \ \dd F(v_1)
        $. And also notice that $c = \mathbb E[v - \theta]^+ = \int_\theta^{V + 1} (v - \theta) \ \dd F(v) =\int_\theta^{V + 1} v \ \dd F(v) - \theta (1 - F(\theta))$. Putting together
        \begin{align*}
            \SW(\theta) & = -cF(\theta) \\
            & \quad + (\int_\theta^{V+1} v\ \dd Fv)(1 + F(\theta)) \\
            & \quad + \int_V^\theta v_2 F(v_2)\ \dd F(v_2) + \int_V^\theta v_1 F(v_1) \ \dd F(v_1)\\
            & = -cF(\theta) + (\int_\theta^{V+1} v\ \dd Fv)(1 + F(\theta)) + 2\int_V^\theta v F(v) \ \dd F(v)\\
            & = \int_\theta^{V + 1} v\  \dd F(v) + \theta F(\theta)(1 - F(\theta)) + 2\int_V^\theta v F(v)\ \dd F(v)
        \end{align*}
        And
        \begin{align*}
            \partial_\theta \SW(\theta) & = -\theta f(\theta) + F(\theta)(1 - F(\theta)) \\
            & \quad + \theta f(\theta) (1 - F(\theta)) - \theta F(\theta)f(\theta) + 2\theta F(\theta) f(\theta)\\
            & = F(\theta)(1 - F(\theta)) \ge 0
        \end{align*}
        When $\theta = V$, $\SW(\theta) = \mathbb E[v]$. So that social welfare can be written in the following simplified form
        $$
        \SW(\theta) = \mathbb E[v] + \int_V^\theta F(s)(1 - F(s)).\ ds
        $$
    \end{proof}
    
\propCSsufficientcondition*

\begin{proof}

\newcommand{\DD}{{\mathcal {D}}} % D(\theta, x)
\newcommand{\g}{{\gamma}} % g(\theta) = D(\theta, x^*(theta))
\newcommand{\DDD}{{\mathbb D}} % D(\theta) = \partial_\theta D(\theta, \rho(\theta))
\newcommand{\h}{{h}} % h(\theta) = \partial_x D(\theta, x) composite x^*(\theta)

    When the inspection cost takes its upper bound $\bar c = \mathbb E[v] - V$, at which buyers simply ignore all non-prominent sellers at any symmetric equilibrium, the equilibrium price can be arbitrarily close to 0 (Proposition \ref{prop:t_star_properties}). 
    In this case, social welfare coincide with the consumer surplus (Proposition~\ref{prop:t_star_properties}, Lemma~\ref{lemma:SW_closed_form}):
    $$
    \SW(\bar c) = \CS(\bar c) = \mathbb E[v].
    $$
    
    By Lemma~\ref{lemma:SW_closed_form}, when the inspection cost $c$ drops in the range $(0, \bar c)$, social welfare rises to
    $$
    \SW(c) = \int_V^{\theta_0(c)} F(v)(1 - F(v))\, \dd v + \SW(\bar c)
    $$
    But the lowest equilibrium price $t^*(c)$ rises as well (Proposition~\ref{prop:t_star_properties}). Starting at $\bar c$, we wish to decrease inspection cost to increase consumer surplus using the prominence mechanism. Lowering inspection cost increases social welfare, meanwhile we would wish the rise in (implementable) equilibrium price lowerbound does not offset the rise in social welfare. Put together the two forces, consumer surplus would be better off if, for inspection cost $c$ that is slightly around the left-neighbourhood of $\bar c$:
    \begin{align}
        \SW(c) - \SW(\bar c) & > t^*(c) - 0\\
        & \Leftrightarrow \\
        \label{expression:consumer_surplus_rising}
        \int_V^{\theta} F(v)(1 - F(v))\, \dd v 
        & > \sup_{x\in (0, \theta - V)} \frac{x\D_c(x)}{\frac12 - \D_c(x)}
    \end{align}
    With a slight abuse of notation, define $\DD:[V, V +1]\times \mathbb R \to \mathbb R$ as:
    $$
    \DD(\theta, x) := \int_V^{\theta_0(c) - x} (1 - F(v + x))\, \dd F(v) \equiv \D_c(x) .
    $$
    Then, the condition for consumer surplus being better off (\ref{expression:consumer_surplus_rising}) can be expressed in the following equivalent form:
    \begin{align}
        & \exists \theta \in (V, V + 1), \exists x \in (0, \theta - V) \text{ such that:}\\
        \label{expression:consumer_surplus_rising_rearranged}
        & x\DD({\theta}, x)+ \varphi(\theta)\DD({\theta}, x) - \frac12 \varphi(\theta) < 0,
    \end{align}
    where, $\varphi(\theta) := \int_V^\theta F(v)(1 - F(v))\, \dd v$.
    To find $\theta, x$ such that (\ref{expression:consumer_surplus_rising_rearranged}) holds, define $\mu :\mathbb R^2 \to \mathbb R$ as
    \label{expression:mu_theta_x}
    \begin{align}
        \mu(\theta, x) := \begin{cases}
            x\DD({\theta}, x) + \varphi(\theta)\DD({\theta}, x) - \frac12 \varphi(\theta), \\
            \quad \quad \quad  \text{for }\theta \in (V, V + 1), x \in (0, \theta - V);\\
            0, \\
            \quad \quad \quad  \text{otherwise}.
        \end{cases}
    \end{align}
    Notice that $\lim_{x \to 0}\mu(\theta, x) = \lim_{x\to \theta - V}\mu(\theta, x) = 0$. Let
    \begin{align*}
        x^*(\theta) & := \begin{cases}
            \arg\min_x \mu(\theta, x) & \theta \in (V, V + 1)\\
            0 & \text{otherwise}.
        \end{cases}
    \end{align*}
    $x^*(\theta)$ is thus a well defined continuous function with $\lim_{\theta \to V} x^*(\theta) = 0$. Let $g(\theta) := \mu(\theta, x^*(\theta))$. Then, $g(V) = 0$. Study the derivative of $g(\cdot)$ when $\theta \to 0^+$.
    By envelope's theorem
    \begin{align*}
        g'(\theta) & = \pder[\mu]{\theta}{(\theta, x^*(\theta))}\\
        & = (x^*(\theta) + \varphi(\theta)) \pder[\DD]{\theta}{({\theta}, {x^*(\theta)})} + \varphi'(\theta)(\DD({\theta}, x^*(\theta)) - \frac12),
    \end{align*}
    and
    \begin{align*}
        \pder[\mathcal D]{\theta}{(\theta, x)}
        & = \pder{\theta}{\int_V^{\theta - x} (1 - F(v + x))\, \dd F(v)},\\
        & = (1 - F(\theta))f(\theta - x)\\
        \varphi'(\theta) & = F(\theta)(1 - F(\theta)),
    \end{align*}
    which jointly implies
    \begin{align*}
        g'(\theta) & = (1 - F(\theta))f(\theta - x^*(\theta))(x^*(\theta) + \varphi(\theta)) \\
        & \quad + F(\theta)(1 - F(\theta))(\DD({\theta}, {x^*(\theta)}) - \frac12)\\
        & = (1 - F(\theta))\left(f(\theta - x^*(\theta))(x^*(\theta) + \varphi(\theta))\right)\\
        & \quad + (1 - F(\theta))\left(F(\theta)(\DD({\theta}, {x^*(\theta)}) - \frac12)\right).
    \end{align*}
    Since $\varphi(V) = 0, x^*(V) = 0$:
    \begin{align*}
        \lim_{\theta \to V^+} g'(\theta) = 0.
    \end{align*}
    Unfortunately, this implies that we need to study higher order derivatives of $g$ and how they behave around $\theta \to V^+$. For convenience, in the subsequent analysis, for any general function $f(\cdot)$, denote $\lim_{x \to V^+} f(x)$ as $f(V^+)$. We first solve $x^*{}'(V^+)$ and $x^*{}''(V^+)$:

    Taking a step back: for general $\rho:[V, V + 1] \to \mathbb R$, define $\DDD_\rho(\theta) := \DD(\theta, \rho(\theta))$, so that:
    \begin{align*}
        \DDD_\rho(\theta) & = \int_V^{\theta - \rho(\theta)}(1 - F(v + \rho(\theta)))\, \dd F(v)
    \end{align*}
    and its first-order derivative:
    \begin{align*}
        \DDD'_\rho(\theta) & = (1 - \rho'(\theta))(1 - F(\theta))f(\theta - \rho(\theta)) \\
        & \quad - \rho'(\theta)\left(\int_V^{\theta - \rho(\theta)} f(v)f(v + \rho(\theta))\, \dd v\right),
    \end{align*}
    and, second-order derivative:
    \begin{align*}
        \DDD''_\rho(\theta) & = (1 - F(\theta))f'(\theta - \rho(\theta))(1 - \rho'(\theta))^2 \\
        & \quad + (1 - F(\theta))f(\theta - \rho(\theta))( - \rho''(\theta))\\
        & \quad - \rho''(\theta)\int_V^{\theta - \rho(\theta)} f(v)f(v + \rho(\theta))\, \dd v\\
        & \quad - \rho'(\theta)\Biggr[(1 - \rho'(\theta))f(\theta - \rho(\theta))f(\theta) \\
        & \quad \quad \quad + \int_V^{\theta - \rho(\theta)} f(v)f'(v + \rho(\theta))\rho'(\theta)\, \dd v\Biggr].
    \end{align*}
    Plug in $\rho = x^*, \theta \to V^+$:
    \begin{align}
    \label{expression:DDD_V_plus_primes_1}
        \DDD'_{x^*}(V^+) & = (1 - x^*{}'(V^+))f(V^+)\\
    \label{expression:DDD_V_plus_primes_2}
        \DDD''_{x^*}(V^+) & = f'(V^+)(1 - x^*{}'(V^+))^2 - f(V^+)(x^*{}''(V^+)) \\
        & \quad - x^*{}'(V^+) (1 - x^*{}'(V^+))f^2(V^+) 
    \end{align}
    
     Recall, $x^*(\cdot)$ is defined as the $x^*$ that achives $\inf_x \mu(\theta, x)$. Since $\mu(\theta, x) = 0 $ for $x\notin (0, \theta - V)$, fix $\theta$, $\mu(\theta, x)$'s minimum is attained at some interior $x \in (0, \theta - V)$, thus for $x^*(\theta) = \arg\min_x \mu(\theta, x)$:
    \begin{align}
        \label{expression:partial_x_D_equals_0}
        \pder[\mu]{x}(\theta, x^*(\theta)) = \pder[\DD]{x}(\theta, x^*(\theta))(x^*(\theta) + \varphi(\theta)) + \DD(\theta, x^*(\theta)) \equiv 0
    \end{align}
    Define $\h(\theta) := \pder[\DD]{x}{(\theta, x^*(\theta))}$, there is
    \begin{align*}
        \h(\theta) & = -(1 - F(\theta))f(\theta - x^*(\theta)) \\
        & \quad - \int_V^{\theta - x^*(\theta)}f(v)f(v + x^*(\theta))\, \dd v\\
        \h'(\theta) & = f(\theta)f(\theta - x^*(\theta)) \\
        & \quad - (1 - F(\theta))f'(\theta - x^*(\theta))(1 - x^*{}'(\theta)) \\
        &\quad  - \Biggr [(1 - x^*{}'(\theta))f(\theta - x^*(\theta))f(\theta) \\
        & \quad \quad \quad + \int_V^{\theta - x^*(\theta)}f(v)f'(\theta - x^*(\theta)) x^*{}'(\theta)\,\dd v \Biggr].
    \end{align*}
    
    Notice that $\h(V^+) = -f(V^+)$, $\h'(V^+) = \frac 12 f^2(V^+) - \frac12 f'(V^+)$. On the basis of notation $\h(\cdot)$ the preceding condition (\ref{expression:partial_x_D_equals_0}) of 
    $\pder[\mu]{x}(\theta, x^*  (\theta))\equiv 0$ is equivalent to
    \begin{align}
        \label{expression:partial_x_D_equals_0_transformed}
        \h(\theta)(x^*(\theta) + \varphi(\theta)) + \DD(\theta, x^*(\theta)) \equiv 0, \forall \theta \in (V, V + 1).
    \end{align}
    Essentially, taking double derivative of (\ref{expression:partial_x_D_equals_0_transformed}) w.r.t. $\theta$ on both side yields
    \begin{align}
    \label{expression:once_der_of_partial_x_D_equals_0_transformed}
        \text{(once) } & \h'(\theta)(x^*{}(\theta) + \varphi(\theta)) \\
        & \quad + \h(\theta)(x^*{}'(\theta) + \varphi'(\theta))  + \DDD'_{x^*}(\theta)= 0\\
    \label{expression:twice_der_of_partial_x_D_equals_0_transformed}
        \text{(twice) } & 2\h'(\theta)(x^*{}'(\theta) + \varphi'(\theta)) \\
        & \quad + \h''(\theta)(x^*(\theta) + \varphi(\theta)) \\
        & \quad + \h(\theta)(x^*{}''(\theta) + \varphi''(\theta)) + \DDD''_{x^*}(\theta) = 0.
    \end{align}
    And $\varphi'(\theta) = F(\theta)(1 - F(\theta))$, $\varphi''(\theta) = f(\theta) - 2F(\theta)f(\theta)$. So, (\ref{expression:once_der_of_partial_x_D_equals_0_transformed}), adding (\ref{expression:DDD_V_plus_primes_1}) simplifies to 
    $$-f(V^+)x^*{}'(V^+) + (1 - x^*{}'(V^+))f(V^+) = 0\\
    \Rightarrow x^*{}'(V^+) = \frac12$$
    Notice that $x^*(V^+) + \varphi(V^+) = 0$, so that we need not further calculate $h''(V^+)$. Integrating $x^*{}'(V^+) = \frac12$ into \ref{expression:DDD_V_plus_primes_2} we obtain
    $$
    \DDD''_{x^*}(V^+) = \frac14 f'(V^+) - f(V^+)x^*{}''(V^+) - \frac14 f^2(V^+)
    $$
    plug in \ref{expression:once_der_of_partial_x_D_equals_0_transformed} solves
    \begin{align*}
        x^*{}''(V^+)
        % & = \frac{-\frac14 f'(V^+) - \frac34 f^2(V^+)}{2f(V^+)}\\
        & = - \frac{f'(V^+) + 3f^2(V^+)}{8f(V^+)}
    \end{align*}

    Back to the function $g$. $g(V^+) = 0 = g'(V^+) = 0$. To study its higher-order derivatives, take $g'(\theta) = (1 - F(\theta))\g(\theta)$, we'd have
    \begin{align*}
        \g(\theta) & := f(\theta - x^*(\theta))(x^*(\theta) + \varphi(\theta)) + F(\theta)(\DD({\theta}, {x^*(\theta)}) - \frac12)\\
        \g'(\theta) & = (1 - x^*{}'(\theta))f'(\theta - x^*(\theta))(x^*(\theta) + \varphi(\theta)) \\
        & \quad + f(\theta - x^*(\theta))(x^*{}'(\theta) + \varphi'(\theta))\\
        & \quad + f(\theta)(\DDD_{x^*}(\theta) - \frac12) + F(\theta)\DDD'_{x^*}(\theta)\\
        \g''(\theta) & = (x^*(\theta) + \varphi(\theta))f''(\theta - x^*(\theta)) \\
        & \quad + 2f'(\theta - x^*(\theta))(1 - x^*{}'(\theta))(x^*{}'(\theta) + \varphi'(\theta))\\
        & \quad + f(\theta - x^*(\theta))(x^*{}''(\theta) + \varphi''(\theta)) + f'(\theta)(\DDD_{x^*}(\theta) - \frac12)\\
        & \quad + 2f(\theta)\DDD'_{x^*}(\theta) + F(\theta)\DDD''_{x^*}(\theta)
    \end{align*}
    Interestingly, $\g'(V^+) = 0$. But
    \begin{align*}
        \g''(V^+) & = \frac12 f'(V^+) + f(V^+)(x^*{}''(V^+) + f(V^+)) - \frac12 f'(V^+)\\
        & = f(V^+)\left(x^*{}''(V^+) + 2f(V^+)\right)\\
        & = \frac18\left(13 f^2(V^+) -  f'(V^+)\right)
    \end{align*}
    Therefore, when $f^2(V^+) < f'(V^+)$, $\exists \theta^\star \in N^+(V)$ such that $g( \theta^\star) < 0$, corresponding to $\mu(\theta^\star, x^*(\theta^\star) < 0 $ --- that starting from its upper bound $\bar c$, decrease in inspection cost (at least, with a small amount) will benefit consumer surplus.
\end{proof}

\end{document}